\documentclass[12pt, draftclsnofoot, journal, letterpaper, onecolumn]{IEEEtran} \IEEEoverridecommandlockouts

\usepackage{amsfonts}
\usepackage[dvips]{graphicx}
\usepackage{times}
\usepackage{cite}
\usepackage{subfigure}
\usepackage{subfig}
\usepackage{caption}
\usepackage{amsmath,amsthm}
\usepackage{array}
\usepackage{amssymb}

\newcommand{\bs}{\boldsymbol}

\usepackage{stfloats}
\usepackage{graphicx}
\usepackage{footnote}
\usepackage{booktabs}
\usepackage{array}

\usepackage[ruled,linesnumbered]{algorithm2e}
 \usepackage{algorithmic}
\usepackage{subeqnarray}
\usepackage{cases}
\usepackage{threeparttable}
\usepackage{color}
\usepackage{url}
\usepackage{caption}
\newcommand{\norm}[1]{\left\lVert#1\right\rVert}

\newtheorem{theorem}{Theorem}

\newtheorem{proposition}{Proposition}
\newtheorem{lemma}{Lemma}

\newtheorem{definition}{Definition}

\newtheorem{question}{Question}
\newtheorem{corollary}{Corollary}

\newtheorem{example}{Example}

\usepackage{mathtools}

\def\BibTeX{{\rm B\kern-.05em{\sc i\kern-.025em b}\kern-.08em
    T\kern-.1667em\lower.7ex\hbox{E}\kern-.125emX}}
    
\ifodd 1

\newcommand{\com}[1]{\textbf{\color{red} (Comment: #1) }}
\newcommand{\comg}[1]{\textbf{\color{blue} (COMMENT: #1)}}
\newcommand{\response}[1]{\textbf{\color{blue} (RESPONSE: #1)}}
\else

\newcommand{\com}[1]{}
\newcommand{\comg}[1]{}
\newcommand{\response}[1]{}
\fi

\ifodd 0
\newcommand{\arx}[1]{#1}
\newcommand{\arxv}[1]{}
\else
\newcommand{\arx}[1]{}
\newcommand{\arxv}[1]{#1}
\fi

\usepackage{mathtools}

\begin{document}

\title{Age-Dependent Differential Privacy}

\author{
 \IEEEauthorblockN{Meng Zhang, \IEEEmembership{Member, IEEE}, Ermin Wei, \IEEEmembership{Member, IEEE},  Randall Berry, \IEEEmembership{Fellow, IEEE}, 
Jianwei Huang, \IEEEmembership{Fellow, IEEE}}
  \vspace{-16pt}
\thanks{
Part of this work has been presented in ACM Sigmetrics/IFIP Performance 2022 \cite{Sigmetrics}.

Meng Zhang is with Zhejiang University, (E-mail: mengzhang@intl.zju.edu.cn).
Ermin Wei and Randall Berry are with Northwestern University (E-mail: ermin.wei@northwestern.edu; rberry@northwestern.edu).
Jianwei Huang is with The Chinese University of Hong Kong, Shenzhen and the Shenzhen Institute of Artificial Intelligence and Robotics for Society. (E-mail: jianweihuang@cuhk.edu.cn).  Corresponding authors are Randall Berry and Jianwei Huang. \\
		This work is supported by the Zhejiang University/University of Illinois at Urbana-Champaign Institute Starting Fund, the Shenzhen Science and Technology Program (Project JCYJ20210324120011032), Guangdong Basic and Applied Basic Research Foundation (Project 2021B1515120008), the Shenzhen Institute of Artificial Intelligence and Robotics for Society, and NSF grant ECCS-2030251.
}
}


     

\maketitle





\begin{abstract}

The proliferation of real-time applications has motivated extensive research on analyzing and optimizing data freshness in the context of \textit{age of information}. 
However, classical frameworks of privacy (e.g., differential privacy (DP)) have overlooked the impact of data freshness on privacy guarantees, which may lead to unnecessary accuracy loss when trying to achieve meaningful privacy guarantees in time-varying databases. 
In this work, we introduce \textit{age-dependent DP}, taking into account the underlying stochastic nature of a time-varying database. 
In this new framework, we establish a connection between classical DP and age-dependent DP, based on which we characterize the impact of data staleness and temporal correlation on privacy guarantees. 
Our characterization demonstrates that \textit{aging}, i.e., using stale data inputs and/or postponing the release of outputs, can be a new strategy to protect data privacy in addition to noise injection in the traditional DP framework. Furthermore, to generalize our results to a multi-query scenario,
we present a sequential composition result for age-dependent DP under any publishing and aging policies. 
We then characterize the optimal tradeoffs between privacy risk and utility and show how this can be achieved.
Finally, case studies show that to achieve a target of an  arbitrarily small privacy risk in a single-query case, combing aging and noise injection only leads to a bounded accuracy loss, whereas using noise injection only (as in the benchmark case of DP) will lead to an unbounded accuracy loss.

\end{abstract}

\section{Introduction}

\subsection{Background}

Fresh data has become indispensable for ubiquitous real-time applications, including Internet-of-things (IoT) systems (e.g., healthcare wearables), cyber-physical systems (e.g., autonomous transportation systems), and financial services. 
For instance, real-time location and velocity information of motor vehicles is the key to 
realize reliable and safe autonomous driving, and real-time stock quotes are essential for developing effective high-frequency trading strategies in financial markets. 
The increasing importance of fresh data has been driving research on a new metric, 
\textit{Age of information (AoI)}, to measure the 
timeliness of the information that a receiver has about the status of a remote source \cite{AoIsurvey}.


An unpredecented amount of personal data is generated in such real-time applications. This may severely compromise user privacy, as 
an adversary may take advantage of this data to infer information about a user. 
This  privacy-sensitive information may include  user identities and user properties, 
whose leakage  may lead to undesireable consequences. For instance,  burglars can use real-time electricity usage readings (in smart grid systems) to infer whether there are people at private homes before attempting burglaries;
stalkers may access real-time GPS location data from location-based service providers 
to track mobile users.




To combat such privacy leakage, researchers have proposed
numerous solutions, varying in terms of the level of data protection and implementation complexity (see a survey in \cite{Survey}).
 A widely used analytical framework is  \textit{differential privacy (DP)} \cite{DP}, which quantifies the level of individual privacy leakage due to releasing aggregate information from a database. 
 Several companies and government agencies, such as  Apple \cite{Aapl}, Google \cite{Goog},  Microsoft \cite{Msft}, and the U.S. Census Bureau \cite{Census}, have adopted the framework of DP in their practices. 
The key idea of DP is to provide strong privacy guarantees by injecting tunable levels of noise 
into the aggregate information before its release, with the goal of 
maintaining a proper  tradeoff between privacy and statistical utility of databases.




Despite its privacy guarantees against arbitrary adversaries, existing DP techniques 
(and other related classes of privacy metrics such as $k$-Anonymity \cite{Anonymity} and $t$-Closeness \cite{Closeness})
have largely overlooked the impact of data freshness for time-varying databases.
Intuitively, as some data has diminishing value over time, releasing outdated data may lead to less privacy leakage if a user only focuses on protecting  its real-time status.
As an example, for a mobile user trying to protect its real-time location, 
the accuracy of an adversary's inference (hence the user's privacy leakage) will significantly drop as the location data becomes outdated. 
In a more concrete example in Section \ref{Sec:Model}, we show that the accuracy of an adversary's estimate diminishes quickly to a level of no privacy concerns due to data aging, whereas the classical DP framework can only provide a very loose upper bound (as high as $88\%$) for the probability an adversary's estimate is correct.

 		\begin{figure*}[t]
		\centering
	\includegraphics[scale=.45]{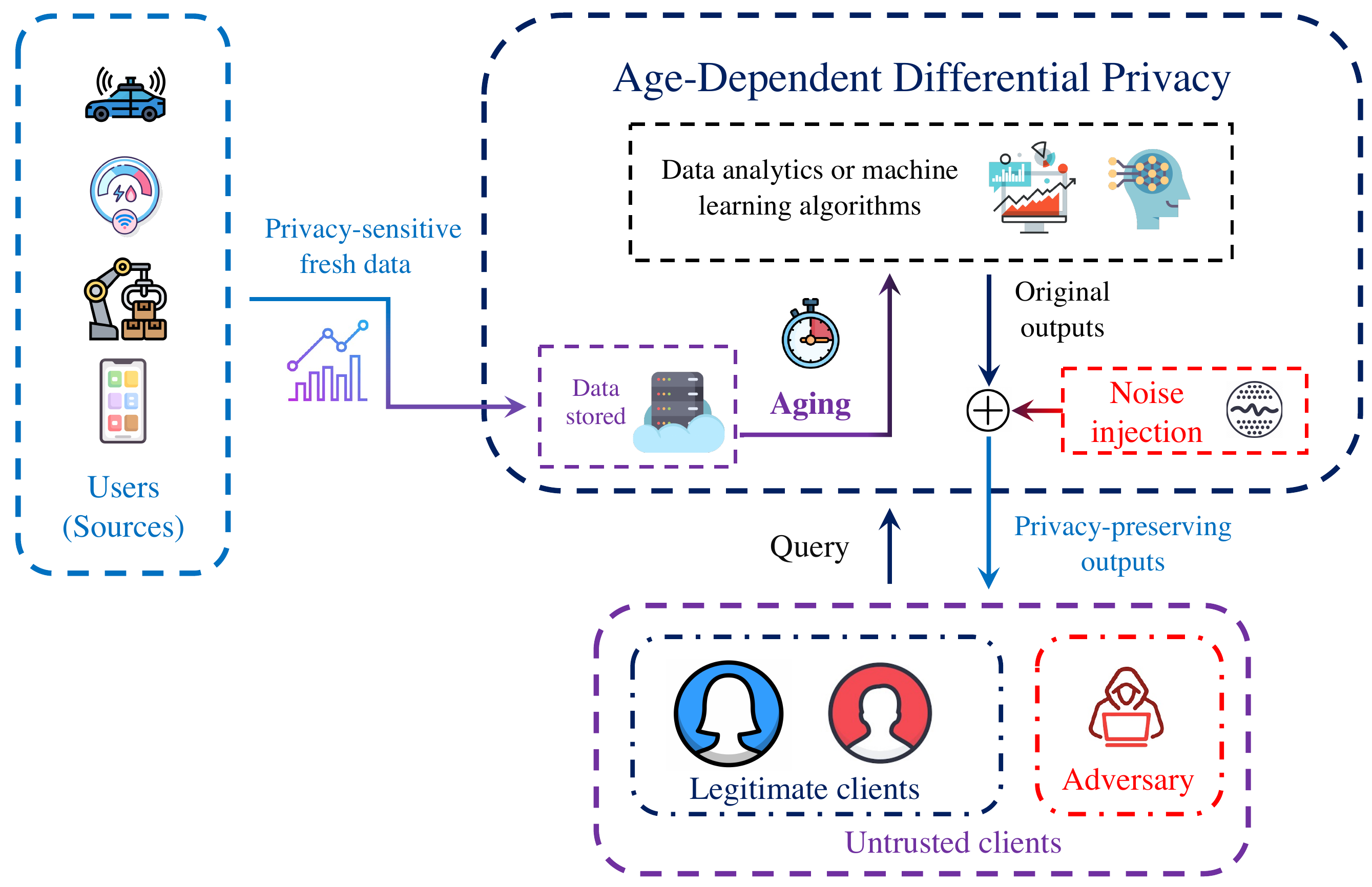}

 		\caption{An age-dependent differential privacy mechanism.}\label{Fig-Sys}
		\vspace{-15pt}
	\end{figure*}

 This observation motivates the following key question:
\begin{question}\label{Q1}
  How should one quantify the impact of data timeliness on data privacy protection?
\end{question}


 Motivated by Question \ref{Q1} and the notions of AoI and DP, this work proposes an age-related generalization of DP that provides more meaningful guarantees for time-varying datasets.
 We name it \textit{age-dependent differential privacy}. 

Our new framework is applicable to a wide range of privacy-sensitive applications with time-varying datasets, such as
(i) location-based services \cite{Location},
(ii) smart meter readings \cite{SmartMeter},
 (iii) medical information, and
 (iv) stock holdings in financial markets.
 Other types of privacy-sensitive information that are time-invariant (e.g., genome information) can be well understood under the classical DP notion, hence are  not the main focus of this paper.

Whereas existing studies largely rely on injecting noise 
to achieve DP \cite{DP}, 
our proposed framework provides a new direction to protect privacy, namely \textit{aging}, which postpones the release of outputs. Figure \ref{Fig-Sys} presents an example of age-dependent differentially private mechanism that combines both the new method (aging) and the existing one (noise injection).
This also raises another challenge 
in designing privacy-persevering mechanisms:
\begin{question}
  How should one characterize age-dependent privacy guarantees by leveraging aging along with classical methods (e.g.,  noise injection)?
\end{question}

To answer this question, we provide theoretic guarantees achieved by classical DP mechanisms when adopted in our new age-dependent DP framework. This establishes a connection between the classical DP notion and our proposed age-dependent generalization. Additionally, it enables us to derive the achievable age-dependent privacy guarantees
by exploiting both classical methods (e.g.,
 noise injection) and aging 
(e.g., timing inputs and outputs).




The operation of real-time systems involving  frequent data updates raises the necessity of 
understanding the performance of privacy-preserving mechanisms with sequential queries, a topic that has been studied for conventional DP mechanisms through work on sequential composition \cite{composition}.
Age-dependent DP mechanisms bring two new challenges compared to these classical results. 
First, our characterization further depends on the \textit{timing} of both inputs (how stale the input database is for all queries) and outputs (when to release the outputs of all queries). In contrast, the classical composition results only depend on the number of queries.
The second challenge is that optimally trading off privacy and utility now depends on   optimizing over such timing as well as any noise that is added, both of which may impact utility and privacy in different ways.
In light of the above challenges, our final key question is
\begin{question}
  How should one characterize the optimal tradeoffs between privacy and utility for multi-query mechanisms?
\end{question}

To this end, we construct multi-query mechanisms by a combination of single-query mechanisms, and derive the corresponding privacy guarantees over time considering the independent guarantees achieved by these single-query mechanisms and the corresponding timings of inputs and outputs. 
By deriving and exploiting a special structure of the optimal solutions, we are then able to formulate a tractable optimization problem that makes tradeoffs between privacy and utility.






\subsection{Contributions} 
We summarize the key contributions of this work in the following:
    \begin{itemize}
      \item \textit{Age-dependent DP.} To the best of our knowledge, age-dependent DP is the first performance metric of privacy protection that accounts for the potentially diminishing privacy leakage of delayed data. 
      \item \textit{Privacy guarantee analysis.} 
    By establishing the connection between the classical DP notion and our proposed age-dependent DP, 
      we derive the achievable privacy guarantees when one exploits aging along with classical DP mechanisms. We also characterize the key factor that determines the decaying rate of age-dependent privacy risks over time. 
      
      \item \textit{Composition.} We further address how age-dependent DP guarantees compose over multiple queries, depending on the timing of both inputs and outputs. Even for infinitely many queries,
      we show that the peak privacy risk can be upper bounded under a characterizable condition.
      
       \item  \textit{Tradeoffs.} 
    We present two case studies to understand tradeoffs between privacy and utility. Our numerical studies reveal that, when approaching an arbitrary small privacy risk in a single-query case, our proposed scheme (of combing aging and noise injection) only incurs a bounded accuracy loss, whereas the loss incurred by a benchmark (using noise injection only) grows unbounded.
    
      
  \end{itemize}
  
  We organize the rest of this paper as follows. In Section \ref{Sec:Related}, we review related studies. Section \ref{Sec:Model} describes the model. In Section \ref{Sec:ADDP}, we introduce the definition of age-dependent DP for a single-query mechanism and characterize its properties. In Section \ref{Sec:Compose}, we generalize our definition and results to multi-query mechanisms and characterize how age-dependent DP guarantees compose. We further present two heuristic examples in Section \ref{Sec:Example} and finally conclude in Section \ref{Sec:Conclusion}.
  
\section{Related Work} \label{Sec:Related}

   \subsection{Age of Information}  
  
  Many works in recent years considered the analysis and optimization for data freshness, measured by AoI, for different networked systems
  (e.g., \cite{AoI1,AoI2,AoI3,AoI4,AoI5,AoI6,AoI7,AoI8,AoI9,AoIPrivacy}
  and a survey in \cite{AoIsurvey}). 
  Most studies in this literature have not considered privacy protection.
  One exception is \cite{AoIPrivacy}, in which Jin \textit{et al.} studied the an age-minimal mobile crowd sensing system for location-based services while achieving DP.
  However, \cite{AoIPrivacy} did not consider an age-based
  variant of DP, nor the impact of data freshness or temporal correlation on privacy protection, which we do here.

  

   \subsection{Differential Privacy}
   


There has been a lot of work on DP (see \cite{DP}) and we will only survey some related studies here. 

\textbf{Variants of differential privacy.} 
The initial work on DP \cite{DP} makes no distributional assumption on the data.
A number of existing variants of DP assumed specific knowledge of data generating distributions, which can provide more meaningful privacy protection  under considered scenarios (e.g., \cite{RDP,BDP1,BDP2,Pufferfish}). For instance, Rinaldo \textit{et al.} in \cite{RDP} studied how data distributions affect the characterizations of privacy guarantees. Kifer \textit{et al.} in \cite{Pufferfish} proposed a Bayesian framework for defining privacy, the Pufferfish. Yang \textit{et al.} in \cite{BDP1} further studied the influence of data correlations on privacy in the Pufferfish framework.
\textit{Our approach assumes the knowledge of an underlying probabilistic model for the time-varying data, whereas these variants mainly did not account for the impact of data timeliness.}


\textbf{Privacy analysis of dynamic databases.} 
Related studies of DP for dynamic databases
have focused on a limited class of problems (e.g., \cite{Dwork2,Growing1,Growing2,Growing3,TimeCorrelation}).
Both \cite{Dwork2} and \cite{Growing2} adapted the notion of DP to streaming environments, where each entry in the database is a single bit, and
bits arrive one per unit time. In \cite{Growing1}, Smith \textit{et al.} extended 
this technique to maintain private sums of real vectors arriving online in a stream. Mechanisms proposed in \cite{Dwork2,Growing1,Growing2} are non-adaptive, in the sense that they answer a single query repeatedly on a dynamic database. Cummings \textit{et al.} extended in \cite{Growing3} the results to an adaptive analysis. 
\textit{Nevertheless, this line of work did not characterize the impact of data timeliness or temporal correlation on data privacy protection.
Reference \cite{TimeCorrelation} is the only exception, in which Cao \textit{et al.} studied the privacy leakage accounting for continuous data release with time correlation (but not data timeliness). 
 The composition result in \cite{TimeCorrelation} did not account for different data publishing and aging policies, which we consider in this work. }




\subsection{Time-Based Privacy Metrics}

Another  line of related work considers time-based metrics for privacy, which only includes a few studies \cite{TimeSuccess,TimeConfuse,TrackTime}.
In \cite{TimeSuccess}, Wright \textit{et al.} proposed a time-based metric that measures the time until the adversary’s success, which in fact assumed that the adversary will eventually succeed.
On the other hand,   Hoh \textit{et al.} in \cite{TimeConfuse} considered \textit{the mean time to confusion}, which measures the time during which the adversary's uncertainty stays below a confusion threshold.
For an adversary aiming to track a target's location over time, Sampigethaya \textit{et al.} in \cite{TrackTime} considered the maximum tracking time to measure the adversary's tracking ability. \textit{These studies did not account for the impacts of the temporal correlation of data or data timeliness on data privacy.}


\section{Model and Preliminary}  \label{Sec:Model}
  
In this section, we introduce the model of the users, the adversary, and the stochastic processes of the users' privacy-sensitive and time-varying databases. After reviewing the classical notion of DP, we use an illustrative example to show that DP may not provide an informative description of privacy protection for time-varying databases. We summarize all key parameters in Table \ref{Tab}.


  \subsection{Model}
  
  \textit{System Overview:} We consider  a set $\mathcal{I}=\{1,\dots, I\}$ of users and an infinite-horizon discrete-time model with time $t\in\mathbb{N}$. The system starts to operate at time $t=0$, and $\mathbb{N}$ denotes the set of all non-negative integers.

\textit{User:} 
For each user $i$, its privacy-sensitive data is captured by 
a discrete-time
 stationary process $\{x_{i,t}\}_{t\in\mathbb{N}}$. 
We assume this process take values in a finite state space,  $\mathcal{X}_i$. Due to the stationarity, there exists a stationary probability distribution $\pi_i: \mathcal{X}_i\rightarrow [0,1]$, such that $\pi_i(x)={\rm Pr}(x_{i,t}=x)$ for all $t\in\mathbb{N}$. We denote the $t$-step transition probability by $P_{i,t}(x,y)\triangleq {\rm Pr}[x_{i,t}=y| x_{i,0}=x]$ for all $t\in\mathbb{N}$ and all $x,y\in\mathcal{X}_i$. 
We further use $\hat{P}_{i,t}(x,y)$ to denote the $t$-step transition probability of the reversed process, which satisfies that
\begin{align}
    \hat{P}_{i,t}(x,y)=\frac{\pi_i(y)P_{i,t}(y,x)}{\pi_i(x)}, \forall x,y\in\mathcal{X}_i, t\in\mathbb{N}, i\in\mathcal{I}.
\end{align}
If $\hat{P}_{i,t}(x,y)={P}_{i,t}(x,y)$ for all $t\in\mathbb{N}$ and $x,y\in\mathcal{X}_i$, then $\{x_{i,t}\}_{t\in\mathbb{N}}$ is \textit{reversible} \cite{Markov}. 

A class of random processes of particular interest is time-invariant Markov chains, which are characterized via
a time-invariant transition probability matrix $P_{i}(x,y)$  for each user $i$ and states  $x$ and $y$.

\textit{Aggregate Database:} 
We use $\{X_t\}_{t\in\mathbb{N}}$ to denote the aggregate process across all users, where $X_t\triangleq\{x_{i,t}\}_{i\in\mathcal{I}}$ is the random variable of
the (aggregate) database of all users at time $t$,
belonging to the (aggregate) state space $\mathcal{X}=\prod_{i\in\mathcal{I}}\mathcal{X}_i$. Throughout this paper, a \textit{database} is referred to as the aggregate data of all users sampled at a time, $X_t$.
Assuming that the processes are independent across users, then the corresponding $t$-step transition probability is $P_{t}(X,Y)=\prod_{i\in\mathcal{I}}P_{i,t}(x_i,y_i)$, the stationary probability is $\pi(X)\triangleq \prod_{i\in\mathcal{I}} \pi(x_i)$, and the $t$-step transition probability of the reversed process is  $\hat{P}_{t}(X,Y)$,
for all $X=\{x_i\}_{i\in\mathcal{I}}$, $Y=\{y_i\}_{i\in\mathcal{I}}$, and $t\in\mathbb{N}$.


\begin{table}[]

\centering
\begin{tabular}{lcc}
 \cline{2-3}  \cline{2-3} 
                                 \multicolumn{1}{l|}{} & \multicolumn{1}{c|}{User $i$'s transition matrix}                                  & \multicolumn{1}{c|}{$P_i(\cdot,\cdot)$}             \\ \cline{2-3} 
\multicolumn{1}{l|}{} & \multicolumn{1}{c|}{User $i$'s $t$-step transition matrix}                         & \multicolumn{1}{c|}{$P_{i,t}(\cdot,\cdot)$}           \\ \cline{2-3} 
\multicolumn{1}{l|}{} & \multicolumn{1}{c|}{User $i$'s $t$-step transition matrix of the reversed process} & \multicolumn{1}{c|}{$\hat{P}_{i,t}(\cdot,\cdot)$}   \\  \cline{2-3}
\multicolumn{1}{l|}{} 
& \multicolumn{1}{c|}{User $i$'s stationary distribution}  &   \multicolumn{1}{c|}{$\pi_i(\cdot)$}                 \\ \cline{2-3} 
\multicolumn{1}{l|}{} & \multicolumn{1}{c|}{Single-query mechanism}                             & \multicolumn{1}{c|}{$M(\cdot)$}                   \\ \cline{2-3} 
\multicolumn{1}{l|}{} & \multicolumn{1}{c|}{Multi-query mechanism}                              & \multicolumn{1}{c|}{$\mathcal{M}_t(\cdot,\cdot)$} \\ \cline{2-3} 
\multicolumn{1}{l|}{} & \multicolumn{1}{c|}{Publishing policy}                                  & \multicolumn{1}{c|}{$\mathcal{S}$}                \\ \cline{2-3} 
\multicolumn{1}{l|}{} & \multicolumn{1}{c|}{Aging policy}                                       & \multicolumn{1}{c|}{$\mathcal{A}$}                \\ \cline{2-3} 
\multicolumn{1}{l|}{} & \multicolumn{1}{c|}{Max. total variation distance}                      & \multicolumn{1}{c|}{$\Delta(t)$}                   \\ \cline{2-3} \multicolumn{1}{l|}{} & \multicolumn{1}{c|}{Classical DP risk} 
                                 & \multicolumn{1}{c|}{$\epsilon_C$}                 \\
                                 \cline{2-3} 
\multicolumn{1}{l|}{} & \multicolumn{1}{c|}{Age-dependent privacy risk}                         & \multicolumn{1}{c|}{$\epsilon(t)$}                \\ \cline{2-3} 
\end{tabular}
\caption{Key Parameters}\label{Tab}
\end{table}

\textit{Event-Level Privacy Protection Goal:} 
Users seek to protect their data privacy against a type of adversaries that aim to infer about the current status $X_t$ at each time slot $t\in\mathbb{N}$. That is, the key focus is to protect each user's single data point at time $t\in\mathbb{N}$, which is referred to as the \textit{event-level} privacy protection \cite{Dwork2,TimeCorrelation}.
Answering a query using database $X_{t'}$ for some $t'<t$ thus leads to less privacy leakage than using database $X_t$, even if $X_{t'}$ and $X_t$ may be correlated.

The consideration of such adversaries is practical in a wide range of application scenarios. 
Examples include real-time electricity consumption data, real-time GPS data, and financial data.
This can be distinguished from  \textit{user-level} privacy protection, where
an adversary is interested in obtaining the entire history (or its subset) of a user's location database $\{X_t\}_{t\in\mathcal{T}}$ for some $\mathcal{T}$, e.g., to determine if a user ever visited a given location with a given time window. We will leave the modeling of this more general class of adversaries and the corresponding analysis for future work.

In our analysis, we consider the worst-case scenario where the adversary has knowledge  of the underlying probability models $\{P_{i,t}(x,y)\}_{i\in\mathcal{I}}$, which provides a strong  guarantee.
 We assume that the platform also has knowledge of $\{P_{i,t}(x,y)\}_{i\in\mathcal{I}}$ and can use it to determine how data can be released.  The platform may gain such knowledge, for example, from analyzing historical data from the same set of users.  
 
 \subsection{Differential Privacy}
 
 Here we review the classical notion of DP \cite{DP}, which can be viewed as if applying to databases that consist of a single element from a joint-state space $\mathcal X$: 
\begin{definition}[Differential Privacy (DP)]\label{DP}
Given any $\epsilon_C >0$, a mechanism (algorithm) ${M}: \mathcal{X}\rightarrow \mathcal{Y}$ is $\epsilon_C$-DP if for all pairs $X,X'\in \mathcal{X}$ which differ only in one user's data, the following inequality holds
\begin{align}
    {\rm Pr}[{M}(X)\in\mathcal{W}]\leq \exp(\epsilon_C)  {\rm Pr}[{M}(X')\in\mathcal{W}], ~\forall \mathcal{W}\subset \mathcal{Y},
\end{align}
where the probability is taken over the randomness of the output of mechanism $M$.
\end{definition}
We use subscript $C$ to indicate that $\epsilon_C$ stands for the classical privacy risk, so as to distinguish it from the privacy risk achievable in our framework.

DP requires that changing each user's data will have little impact on the output when $\epsilon$ is small. 
It provides a strong and mathematically rigorous robustness against an arbitrary adversary, but may come at the expense of the statistical utility (e.g., the accuracy) of the output.

To achieve DP, a well-known approach is to inject Laplace or Gaussian noise to the database output \cite{DP}. Formally, we review the following definitions:
\begin{definition}[Sensitivity]\label{Def2}
For any function $f: \mathcal{X}\rightarrow \mathbb{R}^d$, the $\ell_1$-sensitivity of $f$ is defined as 
\begin{align}
    S^{(f)}=\max_{X,X'}\norm{f(X)-f(X')}_1, \label{Sen}
\end{align}
where $X$ and $X'$ are neighbouring databases, i.e., $X$ and $X'$ differ only in one user's data.
\end{definition}

\begin{definition}[Laplace Mechanism]
For any function $f: \mathcal{X}\rightarrow \mathbb{R}^D$, the Laplace mechanism (associated with $f$) is defined as 
\begin{align}
    \mathcal{M}_L(X)=f(X)+\bs{n}, \label{Laplace}
\end{align}
where $\boldsymbol{n}=\{n_k\}_{k=1}^D$ are independent ${\rm Laplace}(S^{(f)}/\epsilon_C)$ random variables. A random variable has a ${\rm Laplace}(b)$ distribution if its probability density function is
\begin{align}
    p_{\mathcal{L}}(x|b)=\frac{1}{2b}\exp\left(-\frac{|x|}{b}\right).
\end{align}
\end{definition}

\begin{proposition}[Classical Privacy Guarantee \cite{DP}]\label{Prop1}
For any function $f: \mathcal{X}\rightarrow \mathbb{R}^D$ with a sensitivity given in
\eqref{Sen}, the corresponding Laplace mechanism is $\epsilon_C$-DP.
\end{proposition}

As an example, consider function $f(X)=\sum_{i\in\mathcal{I}} X_i/I$, where $X_i \in\{0,1\}$ for every $i\in\mathcal{I}$. From Definition \ref{Def2}, we have $S^{(f)}=1/I$. We use 
the variance of its outcome to characterize the $L_2$ accuracy loss (mean square error) of the Laplace mechanism in \eqref{Laplace},  given by ${\rm Var}(\mathcal{M}_L(X))=2/(\epsilon_C^2 I^2)$. This implies that a higher degree of privacy protection (i.e., a small $\epsilon_C$) requires a larger statistical utility loss (i.e., a larger variance ${\rm Var}(\mathcal{M}_L(X))$).


 Next we present an example to illustrate that the classical DP privacy guarantee may not be practical for time-varying databases, as it does not capture the impact of data freshness on privacy leakage.


\subsection{A Motivating Example}\label{moti}

In this subsection, we consider an illustrative example to motivate our alternative privacy notion.

Consider the following example.\footnote{Our example is the same as the motivational example in \cite{BDP2} except that we consider time-varying databases.} 
There are two towns: $A$ and $B$. The databases $X, X'$ consist of two profiles of resident income values of town $A$.
There is a billionaire $x'$ whose income is orders of magnitude higher than other individuals in either town. 
The billionaire occasionally travels between two towns, and the probability that the billionaire travels  to the other town
after each day is $10\%$.
An adversary wishes to track down the current location of the billionaire. The adversary observes the mean income $w$ of town $A$ sanitised by an $\epsilon_C$-DP mechanism. The result $w$ is only published at time $t=0$. Using a Bayesian estimator, the probability of the adversary being correct is
\begin{align}
    P(A)=\frac{1}{1+\exp(-\epsilon_C)}.\label{AdversaryProb}
\end{align}
For an $\epsilon_C$-DP mechanism with $\epsilon_C = 2$, the attack accuracy is as high as $P(A)=88\%$. 


However, if the adversary obtains this message after $t$ days (still using the same $\epsilon_C$-DP mechanism), we can express the eventual accuracy $P_t(A)$ in the following:
\begin{align}
    \begin{bmatrix}
 P_t(A)  \\
 1-P_t(A)
\end{bmatrix}
=
    \begin{bmatrix}
 P(A)  \\
 1-P(A)
\end{bmatrix}
\cdot \begin{bmatrix}
90\% & 10\% \\
10\% & 90\%
\end{bmatrix}^t.
\end{align}
That is, the probability of a successful attack becomes $P_t(A)\approx 65\%$ when $t=3$ days, $P_t(A)\approx58\%$ when $t=6$ days, and $P_t(A)\approx 53\%$ when $t=10$ days. Therefore, the probability of a successful attack decreases over time and eventually converges to $50\%$  (which corresponds to a pure random guess).

Moreover, the probability of a successful attack when it receives $w$ immediately with $\epsilon_C=0.3$ is approximately the same as  when it observes $w$ after $t=6$ days with $\epsilon_C=2$. 
This shows that aging provides an alternative to increasing the noise variance for increasing privacy protection. 
This motivates us to design an age-dependent generalization that takes the data timeliness into account in the next sections.

\section{Age-Dependent Differential Privacy: The Single-Query Case} \label{Sec:ADDP}

In this section, we start with introducing the definition of \textit{age-dependent differential privacy} for a \textit{single-query mechanism}. In the age-dependent DP framework, we characterize the privacy guarantees achieved by any DP mechanism considering data freshness and the temporal correlation.

\subsection{Definitions}
Throughout Section \ref{Sec:ADDP}, we focus on \textit{single-query (time-invariant) mechanisms}
 $M:\mathcal{X}\rightarrow \mathcal{Y}$, where $\mathcal{X}=\prod_{i\in\mathcal{I}}\mathcal{X}_i$ denotes the state space of an aggregate database (at a specific time instance) and
 $\mathcal{Y}$ is the set of all possible outcomes.  In Section \ref{Sec:Compose}, we will provide an extension to multi-query mechanisms.

    
    
 We first formally introduce age-dependent DP under the single-query  case:
    
\begin{definition}[Age-Dependent Differential Privacy]\label{ADP}
A single-query  mechanism ${M}$ is $(\epsilon,t)$-age-dependent DP for a given random process $\{X_t\}$, if, for any pair ${X},{X}'\in \mathcal{X}$ which differ only in one user's data,  the following is true:
\begin{align}
    {\rm Pr}[{M}(X_{0})\in\mathcal{W}|X_t=X]\leq \exp(\epsilon)  {\rm Pr}[{M}(X_{0})\in\mathcal{W}|X_t=X'], \forall \mathcal{W}\subset \mathcal{Y},\label{Eq-ADP}
\end{align}
where the probability takes into account the randomness of both the output of mechanism $M$ and the stochastic process of $\{X_t\}_{t\in\mathbb{N}}$.
\end{definition}    


Different from classical DP (Definition \ref{DP}), age-dependent DP (Definition \ref{ADP}) also accounts for the evolution of the stochastic process $\{X_t\}_{t\in\mathbb{N}}$.
This also implies that the classical $\epsilon_C$-DP is a special case of the  age-dependent DP, i.e.,
$(\epsilon_C,t=0)$-age-dependent DP.
Due to the time homogeneity of mechanism $M$ and the process $\{X_{t}\}_{t\in\mathbb{N}}$, the inequality \eqref{Eq-ADP} still holds when we replace $X_0$ and $X_t$ in \eqref{Eq-ADP} by $X_{k}$ and $X_{k+t}$, respectively, for any $k\in\mathbb{N}$. 

In \eqref{Eq-ADP}, we interpret $t$  as the \textit{age} of the output $M(X_0)$.
In particular, for a process $\{X_{t}\}_{t\in\mathbb{N}}$ satisfying some \textit{ergodicity} or \textit{mixing} properties (see Appendix \ref{Subsec:Pre}),
 we have that ${\rm Pr}[X_0|X_t]\approx\mu(X_0)$ as $t$ increases (i.e., $X_0$ becomes more outdated).
In other words, as the output data $M(X_0)$ gets more outdated, the probabilities conditional on $X_t=X$ and $X_t=X'$ become indistinguishable and hence mechanism $M$ asymptotically becomes perfectly privacy-preserving.

For any mechanism $M$ achieving $(\epsilon,t)$-age-dependent DP, it is also possible to construct another mechanism that achieves $\epsilon$-DP:
\begin{definition}[Artificially Aging Mechanism]
The aging mechanism $M^t_A$ associated to any arbitrary mechanism $M$ and a given random process $\{X_t\}$, is  given by
\begin{align}
    M^t_A(X)\triangleq M(X_{-t}), 
\end{align}
where $X_{-t}$ is a random database generated based on the reversed process $\hat{P}_{t}$, i.e., $X_{-t}\sim \hat{P}_t(X,\cdot)$. 
\end{definition}
From the definitions of DP and age-dependent DP, we can show that
\begin{corollary}
 An aging mechanism $M^t_A$ is $\epsilon_C$-DP if and only if $M$ is $(\epsilon_C,t)$-age-dependent DP. 
\end{corollary}


\subsection{Properties}

Age-dependent DP satisfies some basic properties of the classical DP, including post-processing resilience (see Appendix \ref{Basic}). In this subsection, we will focus on presenting analytical results to understand how privacy guarantees evolve over time. 
 
\subsubsection{Total Variation Distance} 
To understand the temporal correlation of databases,
we first present an important quantity to measure the statistical distance between two probability distributions.
 \begin{definition}[Total Variation Distance]
  The \textit{total variation distance} between probability distributions $\mu$ and $\pi$ on a finite set $\mathcal{X}$ is defined as
 \begin{align}
    \delta(\mu,\pi) = \max_{A\subset \mathcal{X}}|\mu(A)-\pi(A)|=\frac{1}{2}\sum_{X\in\mathcal{X}}|\mu(X)-\pi(X)|.
 \end{align}
 \end{definition}
To derive age-dependent privacy guarantees for all mechanisms that satisfy the classical DP, we also need the following definition:
\begin{definition}[Maximal Total Variation Distance]
The maximal total variation distance $\Delta(t)$ (between any user's two $t$-step transition probability distributions of the reversed process)  is defined as
\begin{align}
   \Delta(t)\triangleq \max_{i\in\mathcal{I}}\max_{x_{i,0},x_{i,0}'\in\mathcal{X}_i}\delta\left(\hat{P}_{i,t}(x_{i,0},\cdot),\hat{P}_{i,t}(x_{i,0}',\cdot)\right),  \forall t\in\mathbb{N}, \label{MTV}
\end{align}
where $\hat{P}_{i,t}(\cdot,\cdot)$ represents the $t$-step transition probability matrix of user $i$'s reversed process.\end{definition}
We will use $\Delta(t)$ to characterize the privacy risk, i.e., $\epsilon$ in Definition \ref{ADP}.
A wide range of stochastic processes have a diminishing value of $\Delta(t)$ over time at some specific rate. As an example, a wide range of Markov chains satisfy the following 
related and extensively studied property  \cite{Mixing,Geo}:
\begin{definition}[Geometric Ergodicity \cite{Mixing}]\label{GE}
A Markov chain $\{X_{t}\}_{t\in\mathbb{N}}$ with stationary distribution $\pi(\cdot)$ has the geometric ergodicity property if 
\begin{align}
|{\rm Pr}(X_t=Y|X_{0}=X)-\pi(Y)|\leq a(X)\cdot \rho^{-t},~\forall t\in\mathbb{N}, \forall Y\in\mathcal{X},
\end{align}
for some function $a(X)$ and some decay coefficient $\rho<1$.
\end{definition}
In addition to ergodicity, another related notation with a potential property leading to a diminishing value of $\Delta(t)$ is \textit{mixing}. We present details in Appendix \ref{Subsec:Pre}.

\subsubsection{Mechanism-Dependent Guarantee} We start with the following theorem to characterize the formal privacy protection guarantee:
\begin{theorem}[Mechanism-Dependent Guarantee]\label{T2}
If a mechanism $M$ is $\epsilon_C$-DP, then it is also $(\epsilon(t),t)$-age-dependent DP, where $\epsilon(t)$ satisfies
\begin{align}
    \epsilon(t)=\ln\left(1+\Delta(t)
    \cdot (\exp(\epsilon_C)-1)\right),~\forall t\in\mathbb{N}.  \label{eps}
\end{align}
\end{theorem}

 Theorem \ref{T2} characterizes an age-dependent privacy guarantee based on the maximal total variation distance $\Delta(t)$ of the underlying process and the classical DP guarantee $\epsilon_C$. 
The proof of Theorem \ref{T2} mainly involves showing that two transition probability matrices that are close (i.e., $\Delta(t)$ is small) lead to a bounded difference in the likelihoods of the observed value. We defer the complete proof to Section \ref{Subsec:Proofs}.

 
 The significance of Theorem \ref{T2} is two-fold. 
 First, it establishes the connection between the privacy guarantees achieved by classical DP and age-dependent DP. 
 This provides   a 
methodology to attain privacy guarantees in \eqref{eps}, by combining noise injection and aging.
 Second, it indicates that we only need the maximal total variation distance $\Delta(t)$ of the process $\{X_t\}_{t\in\mathbb{N}}$ to characterize such a bound. 
 Note that Theorem \ref{T2} holds for any arbitrary processes $\{X_t\}_{t\in\mathbb{N}}$. When a process has some additional properties, such as satisfying certain mixing or geometric ergodicity (aperiodic and recurrent Markov chains on  finite state spaces) properties, $\Delta(t)$ converges to zero at a certain rate (e.g., a geometric rate). In this case, since $\lim_{x\rightarrow 0}\ln(1+x)/x=1$, \eqref{eps} further implies that the age-dependent privacy risk $\epsilon(t)$ converges to zero at the same rate as $(\exp (\epsilon_C)-1)\Delta(t)$.
 Finally, as we will demonstrate in Section \ref{Sec:Example}, aging along with noise injection
 may not necessarily reduce the data utility  comparing with noise injection alone. 
 
It is also possible to achieve  age-dependent DP without requiring mechanism
 $M$ to be $\epsilon_C$-DP, as we show next.
 \begin{theorem}[Mechanism-Independent Guarantee]\label{T1}
If each agent $i$'s time-varying database is an irreducible, aperiodic, and reversible Markov chain $\{x_{i,t}\}_{t\in \mathbb{N}}$,
any mechanism ${M}$ is $(\epsilon(t),t)$-age-dependent DP satisfying $\epsilon(t)=\mathcal{O}(\lambda_*^t)$. Here, $\lambda_*=\max\{\lambda_1,|\lambda_{m-1}|\}$, where $\lambda_1$ and $\lambda_{m-1}$ are specific eigenvalues of $P$ (see Lemma \ref{L1} in the Appendix).  
\end{theorem}

We present the proof of Theorem \ref{T1} in Appendix \ref{ProofT2}. 
In Theorem \ref{T1}, we use a different proof technique from the one we use in Theorem \ref{T2}, which cancels out the impact of the degree of DP achieved by $M$. 

Theorem \ref{T1} implies that, even without noise injection, using aging alone can ensure an age-dependent privacy guarantee with a similar convergence rate.





\subsection{Proof of Theorem \ref{T2}}\label{Subsec:Proofs}

In this subsection, we present the complete proof of Theorem \ref{T2}.
We first introduce lemmas to characterize the change in expected value when one introduces a small change in probability distribution (based on the maximal total variation distance $\Delta(t)$). We then bound the privacy risk $\epsilon(t)$ based on the privacy risk of DP, $\epsilon_C$, and $\Delta(t)$.

We start with an expected value version of a DP-like probability distribution bound:
\begin{lemma}\label{claim2}
If $M$ is $\epsilon_C$-DP, then for any ${z_i}, z_i'\in\mathcal{X}_i$, we have
\begin{align}
 \mathbb{E}_{\bs{z}_{-i}\sim \hat{{P}}_{-i,t}(\bs{x}_{-i},\cdot)}[{\rm Pr}(\mathcal{M}(z_i,\bs{z}_{-i})\in\mathcal{W})]\leq \exp(\epsilon_C)\mathbb{E}_{\bs{z}_{-i}\sim \hat{{P}}_{-i,t}(\bs{x}_{-i},\cdot)}[{\rm Pr}(\mathcal{M}(z_i',\bs{z}_{-i})\in\mathcal{W})],
\end{align}
where
$\hat{{P}}_{-i,t}(\bs{x}_{-i},\bs{z}_{-i})=\prod_{j\neq i}\hat{P}_t({x}_{j},{z}_{j})$ for all $t\in\mathbb{N}$, $\bs{x}_{-i}$, $\bs{z}_{-i}$, and $i\in\mathcal{I}$.
\end{lemma}
We present the proof in Appendix \ref{Proof-L1}.

The following Lemma \ref{Claim3} characterizes an upper bound for the change in the expected value of $x_i\in\mathcal{X}_i$ when the probability distribution has a small change (characterized by the maximal total variation distance $\Delta(t)$):
\begin{lemma}\label{Claim3}
Assume that every $x_i\in\mathcal{X}_i$ satisfy $~~~\underline{x}_i\leq x_i\leq \bar{x}_i$.  Consider a probability distribution $p(x_i)$ satisfying
$\sum_{x_i\in\mathcal{X}_i} p(x_i)=1$, and the following optimization problem:
\begin{subequations}\label{P}
\begin{align}
    \max_{\bs{\delta}=\{\delta(x_i): x_i\in\mathcal{X}_i\}}\quad &\sum_{x_i\in\mathcal{X}_i}({p}(x_i)+{\delta}(x_i)){x}_i,\\
    {\rm s.t.}\quad &\frac{1}{2}\sum_{x_i\in\mathcal{X}_i}|\delta(x_i)|=\Delta(t), \sum_{x_i\in\mathcal{X}_i}{\delta(x_i)}=0.
\end{align}
\end{subequations}
The optimal solution is $\delta^*(\bar{x}_i)=\Delta(t)$,  $\delta^*(\underline{x}_i)=-\Delta(t)$, and $\delta^*({x}_i')=0$ for all other $x_i'\in\mathcal{X}_i$. The
maximal objective value of \eqref{P} is \begin{align}
 \sum_{x_i\in\mathcal{X}_i}p(x_i) x_i+\Delta(t)(\bar{x}_i-\underline{x}_i).\label{UB}
\end{align}
\end{lemma}
\begin{proof}
We prove Lemma \ref{Claim3} by contradiction.
Suppose that there exists an optimal solution $\bs{\delta}^*$ such that $\delta^*(x_i)\neq 0$ for a $x_i\notin\{\bar{x}_i,\underline{x}_i\}$. We consider the following two cases:
\begin{itemize}
    \item Suppose  that there exists an $x_i'\notin\{\bar{x}_i,\underline{x}_i\}$ such that $\delta^*(x_i')> 0$. We can always construct a new solution $\bar{\bs{\delta}}=\{\delta(x_i):x_i\in\mathcal{X}_i\}$ such that $\bar{\delta}({x}_i')=0$, $\bar{\delta}(\bar{x}_i)=\delta^*(x_i')+\delta^*(\bar{x}_i)$, and $\bar{\delta}(x_i)=\delta^*(x_i)$ for all other $x_i\notin \{\bar{x}_i,{x}_i'\}$. It is easy to check that $\bar{\bs{\delta}}$ is feasible.
    In addition, since $x_i'< \bar{x}_i$, we see $\sum_{x_i\in\mathcal{X}_i}x_i{\delta}^*(x_i)< \sum_{x_i\in\mathcal{X}_i}{x}_i\bar{{\delta}}(x_i)$. Therefore, it contradicts with existence of $x_i'\notin\{\bar{x}_i,\underline{x}_i\}$ such that $\delta^*(x_i')> 0$ at the optimal solution.
    \item Suppose that there exists an $x_i'\notin\{\bar{x}_i,\underline{x}_i\}$ such that $\delta^*(x_i')< 0$. We can always construct a new solution $\bar{\bs{\delta}}=\{\bar{\delta}(x_i): x_i\in\mathcal{X}_i\}$ such that $\bar{\delta}(x_i')=0$, $\bar{\delta}(\underline{x}_i)=\delta^*(x_i')+\delta^*(\underline{x}_i)$, and $\bar{\delta}(x_i)=\delta^*(x_i)$ for all other $x_i\notin \{\underline{x}_i,{x}_i'\}$. It is easy to check that $\bar{\bs{\delta}}$ is feasible. In addition, since $\underline{x}_i< x_i'$, we see $\sum_{i\in\mathcal{I}}{x}_i{\delta}^*(x_i)< \sum_{i\in\mathcal{I}}{x}_i\bar{{\delta}}(x_i)$. Therefore, it contradicts with the existence of $x_i'\notin\{\bar{x}_i,\underline{x}_i\}$ such that $\delta^*(x_i')< 0$ at the optimal solution.
\end{itemize}
Combining the above two cases, we show that all optimal solutions $\bs{\delta}^*$ satisfy $\delta^*(x_i')= 0$ for all $x_i'\notin\{\bar{x}_i,\underline{x}_i\}$. It is readily verified that such an optimal solution should be $\delta^*(\bar{x}_i)=\Delta(t)$, $\delta^*(\underline{x}_i)=-\Delta(t)$, which proves Lemma \ref{Claim3}.
\end{proof}
     
That is, for any probability distribution $p'$ on $\mathcal{X}_i$, such that that total variation distance satisfies $\delta(p,p')\leq \Delta(t)$,
Lemma \ref{Claim3} provides an upper bound of $\sum_{x_i\in\mathcal{X}} p'(x_i)x_i$ in \eqref{UB}.\footnote{In this case, \eqref{UB} is only an upper bound, but not necessarily the same maximal value, as we still need to consider the non-negative constraint that $1\geq p'(x_i)\geq 0$ for all $x_i\in\mathcal{X}_i$, which is  not captured by Lemma \ref{Claim3}.}

In the following, we define $f(z_i)=\mathbb{E}_{\bs{z}_{-i}\sim \hat{P}_{-i,t}(\bs{x}_{-i},\cdot)}[{\rm Pr}(\mathcal{M}(z_i,\bs{z}_{-i})\in\mathcal{W})]$ and $p(z_i)=\hat{P}_{i,t}(x_i',z_i)$ for all $i\in\mathcal{I}$.
We set $\underline{z}_i\triangleq \arg\min_{z_i\in\mathcal{X}_i}f({z}_i)$ and $\bar{z}_i\triangleq \arg\max_{z_i\in\mathcal{X}_i}f({z}_i)$.
 It follows that $1\leq f(z_i)/f(\underline{z}_i)\leq \exp(\epsilon)$ for all $z_i\in\mathcal{X}_i$ from the definition of $\epsilon$-DP.

For all $z_i\in\mathcal{X}_i$, all pairs of $x_i'$ and ${x}_i\in\mathcal{X}$, and all
$\hat{P}_{i,t}$ with the maximal total variation distance given in \eqref{MTV}, it follows that, for all users $i$, 
\begin{align}
&\ln\left(\frac{\sum_{{z}_i\in\mathcal{X}} \hat{P}_{i,t}(x_i,z_i) \sum_{\bs{z}_{-i}\in\mathcal{X}_{-i}}  \prod_{j\neq i} \hat{P}_{i,t}(x_j,z_j) {\rm Pr}[\mathcal{M}(\bs{z})\in\mathcal{W}] }{\sum_{{z}_i\in\mathcal{X}} \hat{P}_{i,t}(x_i',z_i) \sum_{\bs{z}_{-i}\in\mathcal{X}_{-i}}  \prod_{j\neq i} \hat{P}_{i,t}(x_j,z_j){\rm Pr}[\mathcal{M}(\bs{z})\in\mathcal{W}]} \right)\nonumber\\
= &\ln\left(\frac{\mathbb{E}_{z_i\sim \hat{P}_{i,t}(x_i,\cdot)}[f(z_i)]}{\mathbb{E}_{z_i\sim \hat{P}_{i,t}(x_i',\cdot)}[f(z_i)]}
\right)\nonumber\\
\overset{(a)}\leq  &\ln\left(\frac{\sum_{z_i \neq \{\underline{z}_i,\bar{z}_i\}}p(z_i) f(z_i)+(\Delta(t)+p(\bar{z}_i))f(\bar{z}_i)+(p(\underline{z}_i)-\Delta(t))f(\underline{z}_i)}{\sum_{z_i}p(z_i)f(z_i)}
\right)\nonumber\\
\overset{(b)}\leq &\ln\left(\frac{\sum_{z_i \neq \{\underline{z}_i,\bar{z}_i\}}p(z_i) f(z_i)+(\Delta(t)+p(\bar{z}_i))f(\bar{z}_i)+(p(\underline{z}_i)-\Delta(t))f(\underline{z}_i)-\sum_{z_i\in\mathcal{X}_i}p(z_i) (f(z_i)-f(\underline{z}_i))}{\sum_{z_i\in\mathcal{X}_i}p(z_i) f(z_i)-\sum_{z_i\in\mathcal{X}_i}p(z_i) (f(z_i)-f(\underline{z}_i))}
\right)\nonumber\\
=&\ln\left(\frac{\Delta(t) f(\bar{z}_i)+(\sum_{z_i\in\mathcal{X}_i}p(z_i)-\Delta(t))f(\underline{z}_i)}{\sum_{z_i\in\mathcal{X}_i}p(z_i) f(\underline{z}_i)}
\right)\nonumber\\
\overset{(c)}\leq &\ln\left(1+\Delta(t)(\exp(\epsilon_C)-1)
\right),\label{Eq25}
\end{align}
where $(a)$ is from Lemma \ref{Claim3}, $(b)$ is due to the fact that $a/b< (a-c)/(b-c)$ when $a>b>c>0$, and (c) is from Lemma \ref{claim2}.

This completes the proof of Theorem \ref{T2}.

\subsection{Summary}
This section introduced the notion of age-dependent DP. We characterized how the age-dependent privacy risk evolves over time and 
how 
to achieve age-dependent DP by exploiting the classical DP and aging. Our analysis reveals that the key factor of $\{X_t\}_{t\in\mathbb{N}}$ that determines the decaying rate is $\Delta(t)$.
The analysis in this section is only applicable to a single-query  mechanism $M$, whereas the deployment of real-time applications in practice relies on publishing data updates frequently.  This motivates us to study the more general case of multiple queries next.

\section{Age-Dependent Differential Privacy: Sequential Composition} \label{Sec:Compose}

In this section, we introduce a more general notion of age-dependent DP for multi-query mechanisms (e.g., for real-time applications that publish updates frequently) and present a sequential composition theorem. 
We will further discuss how to best trade off privacy and utility (characterized by age of information).

\subsection{The General Definition and Composition}
In a multi-query scenario, the history of all published outputs are visible to the adversary, which brings a new challenge in characterizing the privacy risk over time.
A fundamental problem in the DP literature is how the overall privacy level degrades when combining multiple queries, where each query meets a certain DP guarantee. Such a combination of multiple queries is known as \textit{composition}.
Characterizing classical sequential composition results usually depend on the total number of queries (e.g., \cite{composition}). As an example, Dwork \textit{et al.} in \cite{DP} showed that:
\begin{proposition}[Basic Composition \cite{DP}]\label{P5}
Let $M_n$ be an $\epsilon_{C,n}$-DP mechanism for all $n\in[N]$. The composition ${M}'(X)$ defined as ${M}'(X)=\{M_n(X)\}_{n=1}^N$ is $\epsilon_C$-DP for $\epsilon_C=\sum_{n=1}^N \epsilon_{C,n}$.
\end{proposition}
Proposition \ref{P5} can only provide  a loose privacy guarantee in our case since it does not capture the impacts of data freshness and the temporal correlation.
Another challenge of analyzing time-varying databases comes from the need
to consider timings of both inputs and outputs: \textit{how stale is each input database used for each query} and \textit{when each output is published}. 

To this end, we
consider the following notations of an \textit{aging policy} and a \textit{publishing policy}.  We use $\mathcal{A}=\{A_n\}_{n\in\mathbb{N}}$ to denote the aging policy, where $A_n$ indicates the age of the $n$-th input database used for the $n$-th query.
We further use $\mathcal{S}=\{S_n\}_{n\in\mathbb{N}}$ to denote the publishing policy, i.e., $S_n$ indicates the time instance that the $n$-th outcome is published. Given a publishing policy, we name the time interval $[S_n,S_{n+1})$ as the $n$-th \textit{epoch} for all $n$.


Let  $\mathcal{H}_t$ be the history up to time $t$, given by
\begin{align}
    \mathcal{H}_t(\mathcal{S},\mathcal{A})&\triangleq\{X_{S_n-A_n}: n\in\mathbb{N}, S_n\leq t\},
\end{align}
and let
\begin{align}
N_t&\triangleq \max ~n,\quad {\rm s.t.}\quad S_n\leq t,
\end{align}
be the the number  of queries $N_t$ up to time $t$.
 Furthermore, we define the general class of (potential multi-query) mechanisms as follows:
 \begin{definition}[Multi-query Mechanism]
A multi-query mechanism $\mathcal{M}_t$ is given by
\begin{align}
\mathcal{M}_t(\mathcal{H}_t(\mathcal{S},\mathcal{A}))&\triangleq\{M_n(X_{S_n-A_n}): n\in\mathbb{N}, S_n\leq t\},\label{multiquery}
\end{align}
where $M_n:\mathcal{X}\rightarrow \mathcal{Y}$ is a single-query mechanism used for the $n$-th query.
 \end{definition}

 		\begin{figure*}[t]
		\centering
		\subfigure[]{\includegraphics[scale=.32]{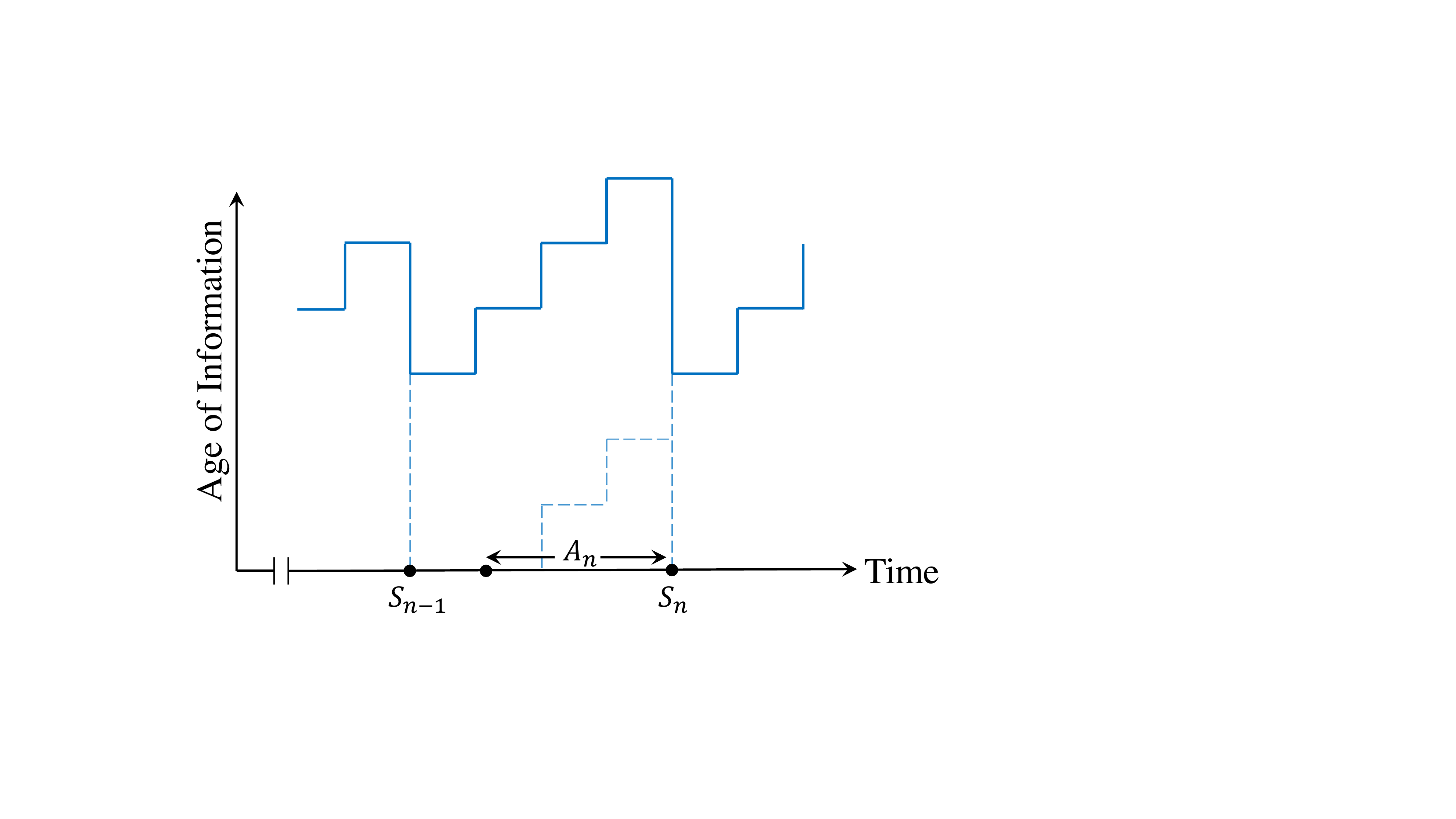}}
		\subfigure[]{\includegraphics[scale=.32]{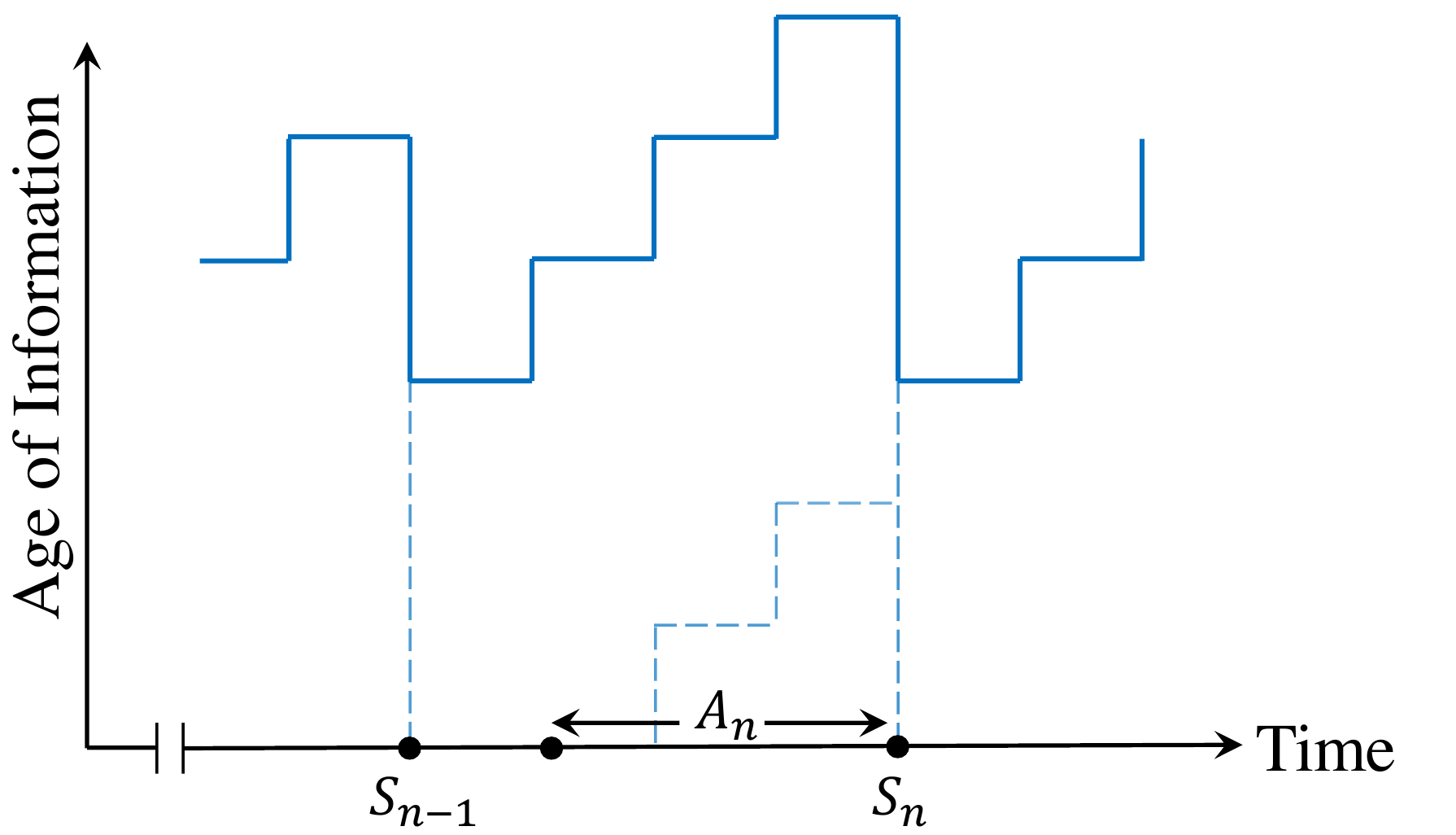}}
		\vspace{-5pt}
		\caption{Illustration of (a) privacy risks of age-dependent DP and (b) age of information for a multi-query mechanism.}\label{Fig1}
		\vspace{-10pt}
	\end{figure*}
A multi-query mechanism is constructed based on potentially infinitely many single-query mechanisms. The $n$-th query uses dataset $X_{S_n-A_n}$ as its input, with $S_n-A_n$ representing the time stamp of the input dataset, and uses (single-query) mechanism $M_n$ and publishes its output at time $S_n$. Further, we  label the multi-query mechanism  with a time stamp $t$ so that $\mathcal{M}_t(\mathcal{H}_t(\mathcal{S},\mathcal{A}))$ stands for the history of all outputs published no later than  time $t$.
	
We next present the general definition of age-dependent DP associated to $\mathcal{M}_t$, characterized by  not only $t$ and $\epsilon$, but also the publishing and the aging policies:
\begin{definition}[Age-Dependent Differential Privacy]\label{M-ADP}
A multi-query mechanism  $\mathcal{M}_t$
 is
$(\epsilon(t),{t},\mathcal{S},\mathcal{A})$-age-dependent DP for a given random process $\{X_t\}$ if for all $t$, the following inequality holds:
\begin{align}
    {\rm Pr}[\mathcal{M}_t(\mathcal{H}_t(\mathcal{S},\mathcal{A}))\in\mathcal{W}_t|X_{t}=X]\leq \exp(\epsilon)  {\rm Pr}[\mathcal{M}_t(\mathcal{H}_t(\mathcal{S},\mathcal{A}))\in\mathcal{W}_t|X_{t}=X'],
\end{align}
for each pair $X,X'\in \mathcal{X}$ which differ only in one user's data, and for all output histories $\mathcal{W}_t\subset\mathcal{Y}^{N_t}$. We let ${S_0}=0$ and $\epsilon(t)=0$ for all $t\in[0,S_1)$.
\end{definition}   

To characterize how the information freshness evolves over time given $\mathcal{S}$ and $\mathcal{A}$, we consider the following definition of age of information:
\begin{definition}[Age of Information \cite{AoIsurvey}]\label{AoI}
Given the publishing policy $\mathcal{S}$ and the aging policy $\mathcal{A}$, we define age of information  ${\rm {AoI}(t)}$ as
\begin{align}\label{Eq-AoI}
    {\rm AoI}(t+1)\triangleq \begin{cases}A_n,
    &~{\rm if }~t+1=S_n,\\
    {\rm AoI}(t)+1, &~{\rm otherwise}.\\
    \end{cases}
\end{align}

\end{definition}
In other words, age of information represents the time elapsed since the time stamp of the input database for the most recently published output.
It has been used to estimate the value (e.g., accuracy) of the output of  $M_n(X_{S_n-A_n})$ \cite{AoIsurvey}.


Before we present the formal characterization of how the age-dependent privacy guarantee evolves over time for any multi-query mechanisms, 
we present an illustrative example of what the age of information in \eqref{Eq-AoI} and the privacy risk guarantees may look like. 
As shown in Figure \ref{Fig1},
the privacy risk $\epsilon(t)$ decreases within each epoch (for any $\Delta(t)$ diminishing in $t$) and spikes at the beginning of the next epoch. Therefore, we define $\epsilon(S_n)$ as the \textit{in-epoch peak privacy risk} for the $n$-th epoch. On the other hand, the age of information increases within each  epoch.

\begin{theorem}[Composition]\label{T3}
For any multi-query  mechanism $\mathcal{M}_t$ given in \eqref{multiquery}, in which
each single-query mechanism $M_n$ is $\epsilon_{C,n}$-DP for all $n\in\mathbb{N}$, and the process $\{X_t\}_{t\in\mathcal{N}}$ is Markovian, then mechanism $\mathcal{M}_t$ is $(\epsilon(t),{t},\mathcal{S},\mathcal{A})$-age-dependent DP, where
\begin{align}
    \epsilon(t)=\ln\left(1+\Delta\left(t-S_{n}+A_{n}\right)\cdot \left(\exp(\epsilon_{C,n}+\epsilon(S_{n}-A_n))-1\right)\right),~\forall t\in[S_n, S_{n+1}), n\in\mathbb{N},\label{Eq-T3}
\end{align}
 where $\Delta(t)$ is given in \eqref{MTV}.
\end{theorem}
We present the proof of Theorem \ref{T3} in Appendix \ref{ProofT3}.
The result in Theorem \ref{T3} has a recurrent form: the privacy risk $\epsilon(t)$ in the $n$-th epoch (i.e., the time interval $[S_n, S_{n+1})$ depends on some previous in-epoch privacy risk, $\epsilon(S_{n}-A_n))$, as well as the privacy risk of the $n$-th (single-query) mechanism, $\epsilon_{C,n}$.

Intuitively,
similar to Theorem \ref{T2}, the maximal total variation distance $\Delta(t-S_n+A_n)$ characterizes the impact of data timeliness on the privacy guarantee. Theorem \ref{T2} generalizes the result in Proposition \ref{P5}, as \eqref{Eq-T3} becomes $\epsilon(t)=\epsilon_{C,n}+\epsilon(S_{n}-A_n)$ when $\Delta\left(t-S_{n}+A_{n}\right)=1$ (as in a static database).
This result in fact coincides with the basic composition result (for static databases).
On the other hand, several existing studies have provided state-of-art improvements on composition (e.g., \cite{composition,Growing3}) by considering more sophisticated composition techniques (e.g. adaptive composition as in \cite{Growing3}). Hence, it may also be possible to obtain a tighter bound than \eqref{Eq-T3} by combining sophisticated composition techniques and the impact of data timeliness, which will be left for future work.

\subsection{Tradeoff Between Privacy and Utility}\label{ESP}

Theorem \ref{T3} only characterizes a privacy risk for each specific time.
To facilitate our analysis of tradeoffs, we need another metric to capture the overall privacy protection performance over the entire time horizon. We thus consider the (overall) peak privacy risk:\footnote{We note that the literature of age of information has extensively used \textit{peak age of information} as their overall performance metric as well. \cite{AoIsurvey}.}
\begin{align}
    \epsilon^*\triangleq \sup_n\epsilon(S_n). \label{Z5}
\end{align}
To make the optimal tradeoff between the privacy and utility, we let $\bar{S}_n\triangleq S_{n+1}-S_n$ denote the inter-publish time between the $n$-th and the $(n+1)$-th updates, for all $n\in\mathbb{N}$.

We further define a (noise-aware) peak age penalty, denoted by
 \begin{align}
     f(\max_{t}{\rm AoI}(t),\epsilon_C)=f(A+\bar{S},\epsilon_C),
 \end{align}
 where $f(t,\epsilon)$ is increasing in $t$ and decreasing in $\epsilon_C$. Function $f$ measures the accuracy loss of $M_t(\mathcal{H}_t(\mathcal{S}^E,\mathcal{A}^E))$ due to data staleness (characterized by the peak age of information $A+\bar{S}$) and the injected noise (characterized by $\epsilon_C$).

We consider the following optimization problem:
\begin{subequations}\label{GeneralProblem}
\begin{align}
    \min\quad&\sup_{n\in\mathbb{N}} \epsilon(S_n)\\
    \rm{s.t.}\quad&f(A_n+\bar{S}_n,\epsilon_{C,n})\leq \bar{f},\forall n\in\mathbb{N},\label{constraint:age}\\
    \rm{var.}\quad& \epsilon_n\geq 0, {A}_n\geq 0, \bar{S}_n\geq 0, \forall n\in\mathbb{N},
\end{align}
\end{subequations}
where \eqref{constraint:age} stands for an peak age penalty constraint.
The challenge in solving Problem \eqref{GeneralProblem} mainly lies in the difficulty of dealing with the recurrent form of $\epsilon(S_n)$ instead of a closed form.
To optimally solve Problem \eqref{GeneralProblem}, we will
show in the following that  a specific class of solutions $(\mathcal{S},\mathcal{A},\bs\epsilon)$ are optimal. 
We can then express $\sup\epsilon(S_n)$  in a closed form.

Specifically, we consider
the following simplified mechanism, defined as:
\begin{definition}[Simplified Multi-Query Mechanism]
A simplified multi-query mechanism $(\mathcal{S}^E=\{S_n^E\}_{n\in\mathbb{N}},\mathcal{A}^E=\{A_n^E\}_{n\in\mathbb{N}})$ is a mechanism in which 
\begin{align}
    A^E_n=A,\quad S^E_{n+1}-S^E_n=\bar{S}, ~{\rm and}~\epsilon_{C,n}=\epsilon_C, ~\forall n\in\mathbb{N},\label{SE}
\end{align}
for some ($A,\bar{S}, \epsilon_C$). 
\end{definition}


In other words, a simplified mechanisms is one in which updates are published at regular intervals and every update is aged the same amount and sent using the same single-query DP mechanism. By exploiting that the monotonicity of $\epsilon(S_n)$
in $n$ and the existence of an optimal solution that have identical values of  $(A_n^*,\bar{S}_n^*, \epsilon_{C,n}^*)$, we can prove the existence of an optimal solution satisfying \eqref{SE} which leads to the following theorem:
\begin{theorem}\label{T4}
There exists an optimal solution to \eqref{GeneralProblem} that is a simplified multi-query mechanism, i.e., $A_n^*=A^*$, $\bar{S}_n^*=\bar{S}^*$ and $\epsilon_{C,n}^*=\epsilon_{C}^*$ for some $(A^*,\bar{S}^*,\epsilon_{C}^*)$.
\end{theorem}

We present the complete proof of Theorem \ref{T4} in Appendix \ref{ProofT4}. Theorem \ref{T4} indicates that there is no need to vary polices across different epoches, and hence we only need to focus on optimizing three variables $(A,\bar{S},\epsilon_C)$. 


 		\begin{figure*}[t]
		\centering
	\includegraphics[scale=.32]{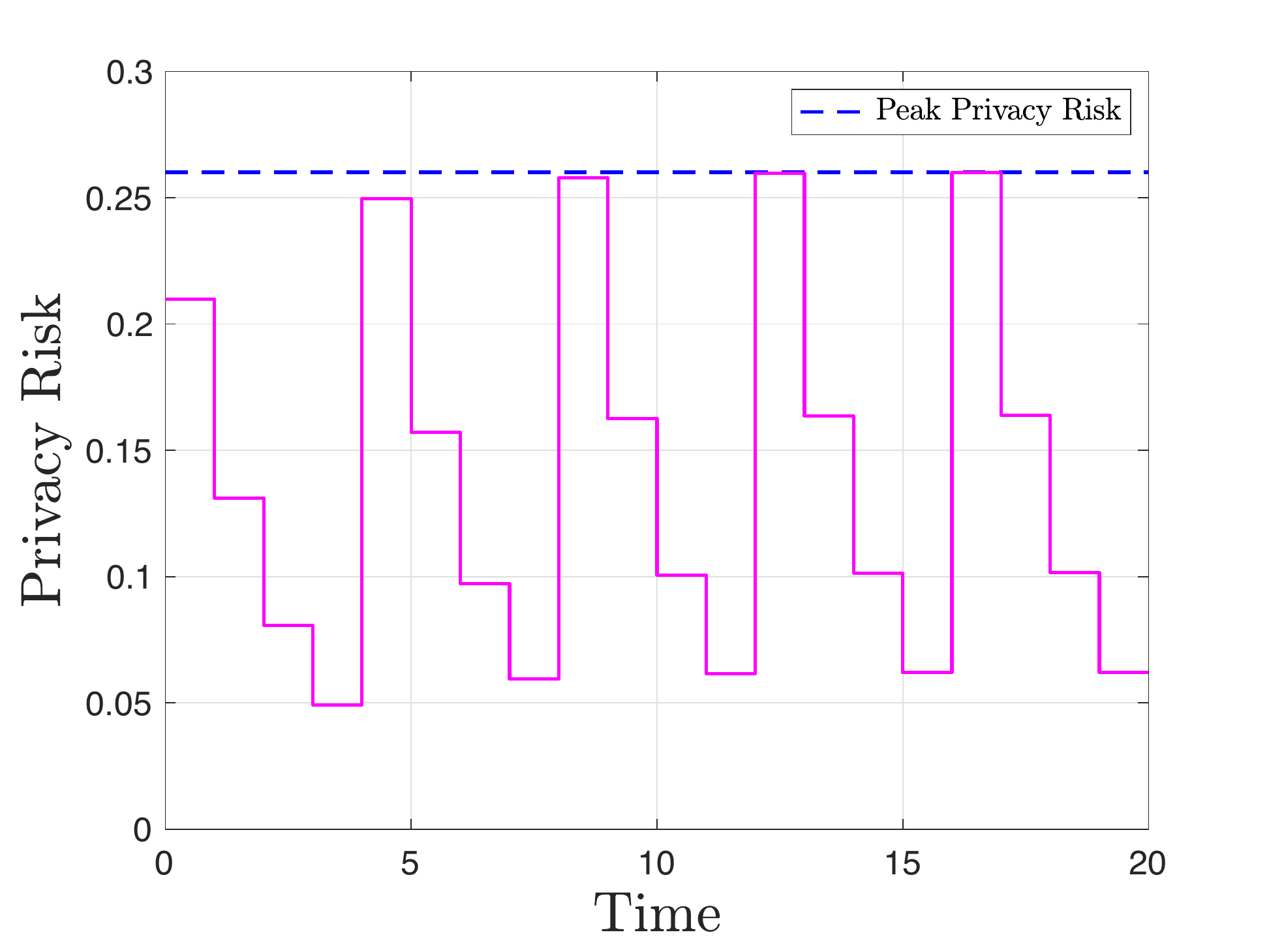}
		\vspace{-5pt}
 		\caption{An illustrative example of the privacy risk $\epsilon(t)$ and the peak privacy risk 
 		for equal-spacing policies. We set $A=2$,  $\bar{S}=4$, and $\epsilon_C=0.5$.}\label{Fig2}
		\vspace{-10pt}
	\end{figure*}
	
In light of Theorem \ref{T4}, the following shows the peak privacy risk in  \eqref{Z5} for a simplified multi-query mechanism:
\begin{proposition}[Peak Privacy Risk] \label{P6}
If the process $\{X_t\}_{t\in\mathbb{N}}$ is Markovian, for
a simplified multi-query mechanism $\mathcal{M}_t(\mathcal{H}_t(\mathcal{S}^E,\mathcal{A}^E))$ with an equal-spacing policy $(\mathcal{S}^E,\mathcal{A}^E)$ and each single-query mechanism $M_n$ being $\epsilon_{C}$-DP for all $n\in\mathcal{N}$, the peak privacy risk is given by the unique fixed point satisfying
\begin{align}
    \epsilon^*=\ln \left(1+\frac{\Delta(A)\cdot(\exp(\epsilon_C)-1)}{1-\Delta(A)\Delta(\bar{S}-A)\exp(\epsilon_C)}\right), \label{Z6}
\end{align}
 whenever $\Delta(A)\Delta(\bar{S}-A)\exp(\epsilon_C)<1$. 
 Furthermore, the following fixed-point iteration constitutes a contraction mapping:
  \begin{align}
    \epsilon(S_n)=\ln\left(1+\Delta(A)\cdot \left(\exp(\epsilon_C)\exp(\epsilon(S_{n-1}-A)))-1\right)\right).\label{inepo}
\end{align}
\end{proposition}
From \eqref{Z6}, we note that $\Delta(A)\Delta(\bar{S}-A)\exp(\epsilon_C)<1$ is the condition for the existence of a finite value of $\epsilon^*$.
That is, when the inter-publishing time $\bar{S}$ is sufficiently long and  the privacy leakage for each query $\epsilon_C$ is sufficiently small, then such a unique fixed point in \eqref{Z6} exists. Otherwise, the privacy risk $\epsilon(S_n)$ diverges to infinity.

\begin{proof}
The fixed point expression in \eqref{Z6} comes directly from  \eqref{Eq-T3} and \eqref{Z5}.
In addition, define $$F(\epsilon)\triangleq\ln\left(1+\Delta(A)\cdot \left(\exp(\epsilon_C)\exp(\epsilon))-1\right)\right).$$ 
It follows that
\begin{align}
    0< \frac{d F(\epsilon)}{d \epsilon}=\frac{\Delta(A)\cdot \exp(\epsilon_C)\exp(\epsilon)}{\Delta(A)\cdot \left(\exp(\epsilon_C)\exp(\epsilon)-1\right)+1}<1,
 \end{align}
 which implies that $F(\epsilon)$ has a Lipschitz constant $0<L<1$. Therefore, \eqref{inepo} is a contraction mapping.
\end{proof}
We present an illustration in Figure \ref{Fig2}. We observe that the in-epoch peak privacy risk $\epsilon(S_n)$ increases in $n$ and converges to the fixed point in \eqref{Z6}. 
It only takes $5$ epoches for the privacy risk to approximately attain the fixed point.
 
 \subsection{Optimization Algorithms}
 In the following, we consider an optimization problem that accounts for tradeoffs  between privacy and utility. 
 Since many practical stochastic processes may have a geometrically decaying $\Delta(t)$ as we have shown in Appendix \ref{Subsec:Pre} and Section \ref{Sec:ADDP}, we are motivated to focus on the following specific form:
 \begin{align}
     \Delta(t)= c\cdot\rho^t, \forall t\in\mathbb{N},
 \end{align}
 for coefficients $\rho\in(0,1)$ and $c\geq 1$.

 
 Instead of \eqref{GeneralProblem}, we consider the following peak age risk minimization problem:
 \begin{subequations}\label{Optimization}
 \begin{align}
    \min\quad &\ln \left(1+\frac{c\rho^{A}(\exp(\epsilon_C)-1)}{1-c^2\rho^{\bar{S}}\exp(\epsilon_C)}\right)\label{Eq-Opt-a}\\
    {\rm s.t.}~\quad & c^2\rho^{\bar{S}}\exp(\epsilon_C)<1,\label{Eq-Opt-b}\\
    & f(A+\bar{S},\epsilon_C)\leq \bar{f}, \label{Eq-Opt-c}\\
    {\rm var.}\quad & \epsilon_C\geq 0, A\geq 0, \bar{S}\geq 0,
\end{align}
 \end{subequations}
where the constraint in \eqref{Eq-Opt-c} indicates that the peak age-of-information penalty is upper-bounded by a threshold $\bar{f}$. 
Note that we can drop constraint $\bar{S}\geq 0$ because \eqref{Eq-Opt-b} implies that $\bar{S}$ must be positive.
For trackability in \eqref{Optimization}, we relax the integer constraints on $A$ and $\bar{S}$. After obtaining the optimal solution ($\epsilon_C^*,A^*,\bar{S}^*$) to \eqref{Optimization}, we can round $A^*$ and $\bar{S}^*$ to their respective nearest integers to obtain an approximate solution.

\begin{proposition}\label{P7}
When the constraint, $A\geq 0$, is not binding,
the optimal solution to Problem  \eqref{Optimization} satisfies
\begin{align}
    \rho^{\bar{S}}\exp(\epsilon_C)=\frac{1}{2c^2}, \label{QS}
\end{align}
or equivalently, $\ln(\rho^{-1})\bar{S}=\ln(2c^2)+\epsilon_C$. 
\end{proposition}
This shows that when aging is used ($A>0$), then as $\epsilon_C$ decreases (more privacy due to adding noise) then $\bar{S}$ also decreases meaning that data is published more often. We prove Proposition \ref{P7} by exploiting the Karush–Kuhn–Tucker conditions of Problem \eqref{Optimization}, and present the proof  in Appendix \ref{ProofP7}.


	\begin{algorithm}[t]
\caption{Heuristic Solution to \eqref{Optimization-2}}\label{Algo1}

Initialize $\phi$, $K$, and $\bar{\epsilon}$;\
    \For{$k=\{1,2,...,K\}$ }{
          Set $\epsilon_C(k)=\frac{\bar{\epsilon}k}{K}$ and
        $\bar{S}(k)=(\ln(2c^2)+\epsilon_C(k))/\ln(\rho^{-1})$\; \tcp{Exhaustive search for $\epsilon_C$}
       Set $A_L=0$ and $A_H=\tilde{A}$, where $\tilde{A}$ satisfies that $\lim_{\epsilon_C\rightarrow\infty}f(\tilde{A},\epsilon_C(k))=\bar{f}$\;
       \While{$A_H-A_L\leq \phi$} {\tcp{Bisection search for $A$}
 \eIf{$f(A+\bar{S},\epsilon_C(k))\leq \bar{f}$}{
         Set $A_L=A(k)$\;}{
       
         Set $A_H=A(k)$\;}
       
    } \label{roy's loop}
    
   }
 Select $k^*=\arg\min \rho^{A(k)}\cdot(\exp(\epsilon_C(k))-1)$\;
 \Return the solution $\{\epsilon_C(k^*),A(k^*),\bar{S}(k^*)\}$.
 

\end{algorithm}

In light of Proposition \ref{P7}, we solve the reduced optimization problems of \eqref{Optimization} in the following two cases, depending on whether constraint $A\geq 0$ in \eqref{Optimization} is binding or not:

\subsubsection{When the constraint $A\geq 0$ is not binding} Substituting \eqref{QS} into Problem \eqref{Optimization}, we have the following equivalent reformulated problem:
 \begin{subequations}\label{Optimization-2}
 \begin{align}
    \min\quad & \rho^{A}\cdot(\exp(\epsilon_C)-1)\label{Opt-2-a}\\
    {\rm s.t.}\quad & \ln(\rho^{-1})\bar{S}=\ln(2c^2)+\epsilon_C,\\
    & f(A+\bar{S},\epsilon_C)\leq \bar{f},\label{Z}\\
     {\rm var.}\quad & \epsilon_C\geq 0, A\geq 0, \bar{S}\geq 0.
\end{align}
 \end{subequations}
Problem  \eqref{Optimization-2} is a convex problem whenever $f(t,\epsilon_C)$ is convex in $(t,\epsilon_C)$, in which case Problem \eqref{Optimization-2} can be readily solved by standard solvers (e.g., CVX \cite{CVX}). When $f(t,\epsilon_C)$ is non-convex, we present a heuristic algorithm to solve Problem \eqref{Optimization-2} as shown in Algorithm \ref{Algo1}. The computational complexity of Algorithm \ref{Algo1} is $\mathcal{O}(K \log(1/\phi))$, where $K$ comes from the exhaustive search for $\epsilon_C$ and $\log(1/\phi)$ comes from the bisection search for $A$.

\subsubsection{When the constraint $A\geq 0$ is binding} Substituting $A=0$ into Problem \eqref{Optimization}, we have the following equivalent reformulated problem:
 \begin{subequations}\label{Optimization-3}
 \begin{align}
    \min\quad & \frac{\exp(\epsilon_C)-1}{1-c^2\cdot \rho^{\bar{S}}\exp(\epsilon_C)}\label{Opt-3-a}\\
    {\rm s.t.}~~\quad & c^2\cdot\rho^{\bar{S}}\exp(\epsilon_C)<1,\\
    & f(\bar{S},\epsilon_C)\leq \bar{f},\label{Z2}\\
     {\rm var.}\quad & \epsilon_C\geq 0, \bar{S}\geq 0.
\end{align}
 \end{subequations}

	\begin{algorithm}[t]
\caption{Heuristic Solution to \eqref{Optimization-3}}\label{Algo2}
\SetAlgoLined
 Initialize $\phi$, $K$, and $\bar{\epsilon}$\; 
 
    \For{$k\in\{1, 2, ..., K\}$ }{\tcp{Exhaustive search for $\epsilon_C$}Set $\epsilon_C(k)=k\bar{\epsilon}/K$\;
        Set $\bar{S}_L=0$ and $\bar{S}_H=\bar{B}$, where $\bar{B}$ satisfies that $\lim_{\epsilon_C\rightarrow\infty}f(\bar{B},\epsilon_C)=\bar{f}$
        \;
    \While{$\bar{S}_H-\bar{S}_L\leq \phi$}{ \tcp{Bisection search for $\bar{S}$}
 \eIf{$f(\bar{S},\epsilon_C)\leq \bar{f}$}{
         Set $\bar{S}_L=\bar{S}(k)$\;}{
        Set $\bar{S}_H=\bar{S}(k)$\;}
         }
         
    
   }
Select $k^*=\arg\min \frac{\exp(\epsilon_C(k))-1}{1-c^2\cdot \rho^{\bar{S}(k)}\exp(\epsilon_C(k))}$\;
 \Return the solution $\{\epsilon_C(k^*),A=0,\bar{S}(k^*)\}$;
\end{algorithm}

We solve Problem \eqref{Optimization-3} in Algorithm \ref{Algo2},  in which we set $A=0$ and search for the optimal $\epsilon_C$ exhaustively and use a bisection search for the optimal $\bar{S}$. Specifically, for each $\epsilon_C$, we use a bisection method to search for the corresponding $\bar{S}$ such that $f(\bar{S},\epsilon_C)=\bar{f}$. We then select $\epsilon_C$ and the corresponding optimal $\bar{S}$ to attain the minimal value of the objective in \eqref{Opt-3-a}. Similarly, the complexity of Algorithm \ref{Algo2} is also $\mathcal{O}(K \log(1/\phi))$.


\subsection{Possibility of No Tradeoffs}

Intuitively, increasing privacy will lead to a decrease of utility (measured by an increase in the peak age penalty function $f$).  However, unlike the single-query mechanism, this tradeoff between privacy and utility in the multi-query scenario  may not exist, i.e., there may be cases where increasing the bound on the penalty ($\bar{f}$) does not lead to greater peak privacy risk. We provide a sufficient condition for the  non-existence of a tradeoff in the following corollary:

\begin{corollary}\label{Coro2}
If $f\left(\frac{\ln(\exp(\epsilon_C)-1)+\epsilon_C}{\ln(\rho^{-1})}+a,\epsilon_C\right)$ is monotonically decreasing in $\epsilon_C$ for any coefficient $a$, then an increase in $\bar{f}$ does not decrease the minimal objective value in Problem 
\eqref{Optimization}.
\end{corollary}

We present the proof of Corollary \ref{Coro2} in Appendix \ref{Proof-Coro2}. The main idea of the proof involves showing the possibility of constructing a new solution $f(A'+\bar{S}',\epsilon_C')<f(A+\bar{S},\epsilon_C)$ while maintaining the peak privacy risk under the condition in Corollary \ref{Coro2}. Corollary \ref{Coro2} implies that sacrificing one of the metrics (privacy or utility) does not necessarily improve the performance of another. In Section Section \ref{Sub-Sec:Two-State}, we will present concrete examples that satisfy and violate the condition in Corollary \ref{Coro2}.


\section{Numerical and Experimental Examples}\label{Sec:Example}

In this section, we consider 
two concrete examples of modeling the process $\{X_t\}$. We first consider  a two-state Markov chain, to study the tradeoffs between age-dependent DP and accuracy loss achieved by our proposed scheme (by aging and noise injection) and the classical scheme. 
We further consider an experiment of our proposed scheme applied to electricity consumption data.
We also consider an autoregressive model and show how our analysis can also be extended to continuous-valued processes in Appendix \ref{AGM}.



\subsection{Two-State Markov Chain}\label{Sub-Sec:Two-State}
In this subsection, we first consider the following example with two states \cite{Pufferfish2}:
\begin{example}\label{eg1}
Consider a system of  $I$ agents and two locations. We use $x_{i,t}\in\{-1,1\}$ to denote the location at time $t$ of agent $i$ and let $X_t=\{x_{i,t}\}_{i\in\mathcal{I}}$.
Each agent $i$'s time-varying database is given by an identical discrete-time Markov chain, with the following transition probability matrix\footnote{We drop the index $i$ in the matrices for Example \ref{eg1}, as agents' time-varying databases are identical.}:
\begin{align}
    P=\begin{pmatrix}
1-p & p \\
q & 1-q
\end{pmatrix},
\end{align}
which can be verified to be reversible.
Assuming $|1-p-q|<1$,\footnote{If $p=q=0$, the Markov chain is reducible. If $p=q=1$, the Markov chain is periodic.} it follows that the stationary distribution is $\pi(0)=\frac{q}{p+q}$ and $\pi(1)=\frac{p}{p+q}$, and the $t$-step  transition probability matrix is:
\begin{align}
    {P}_t&=\frac{1}{p+q}\begin{pmatrix}q & p \\ q & p\end{pmatrix}+ \frac{(1-p-q)^t}{p+q}\begin{pmatrix}p & -p \\ -q & q\end{pmatrix}.\label{Eq-TS}
 \end{align}
It follows that 
\begin{align}
        \norm{P_{t,-1}-\pi}&= \frac{p}{p+q}|1-p-q|^t,\\ \norm{P_{t,1}-\pi}&= \frac{q}{p+q}|1-p-q|^t,
\end{align}
where $P_{t,x}(y)=P_t(x,y)$ for all $x$, $y\in\{-1,1\}$.
Therefore, the maximal total variation distance satisfies
\begin{align}
    \Delta(t)&=|1-p-q|^t, \forall t\in\mathbb{N}.\label{Eq40-Del}
\end{align}
\end{example}
Such a two-state Markov chain model may capture several practical scenarios, e.g., i) the status of a home being empty or not, ii) a user being on a road at one of two locations.
In the latter example, an analyst aims at designing a mechanism to estimate the current aggregate traffic congestion levels of two locations (e.g., to estimate the traffic conditions), whereas users may try to preserve their own location privacy.

We consider a common single-query aggregation mechanism ${M}$ that estimates the average value of $x_{i,t}$ across all users, e.g., to analyze the overall traffic conditions:
\begin{align}
    M(X_0)=\frac{1}{I}\sum_{i\in\mathcal{I}}x_{i,0}+n,\label{outcome}
\end{align}
where $n$ is a ${\rm Laplace}(\frac{1}{\epsilon_C I})$ random variable (which ensures  $M$ to be $\epsilon_C$-DP by Proposition \ref{Prop1}).
Based on \eqref{outcome}, we consider two utility metrics:
\begin{itemize}
    \item \textit{Mean-Square Error}: We assume no prior information is available and simply use the given released data $M(X_0)$ as the estimate.
The mean-square error (MSE) is given by\cite{MMSE}:
\begin{align}
        f_{\rm MSE}(t,\epsilon_C)= \mathbb{E}\left[\left(M(X_0)-\frac{1}{I}\sum_{i}x_{i,t}\right)^2\right].
\end{align}
Such a mean-squared (estimation) accuracy loss is a commonly used age penalty function  (e.g., \cite{AoIsurvey,AoI0}).

\item \textit{Failure Rate}: Motivated by the age of incorrect information proposed in \cite{incorrect}, 
we consider a binary estimator. Specifically, based on the outcome of the mechanism $M(X_0)$, we 
seek to determine  whether $\sum_{i}x_{i,t}/I\geq 0$ is true (e.g., whether a the traffic at a specific location is congested). 
Let $z_t$ be an indicator random variable for the event that $\sum_i x_{i,t}/I \geq 0$. 
The maximum likelihood estimate of $z_t$ is given by:
\begin{align}
    g_{\rm MLE}(M(X_0)) \triangleq \arg\max_{x\in\{0,1\}}{\rm Pr}[z_{t}=x|{M}(X_0)].
\end{align}
We use the following definition of the failure rate (conditional on the initial state $X_0$) as our noise-aware age penalty:
\begin{align}
    f_{\rm MLE}(t, \epsilon_c)={\rm Pr}[  g_{\rm MLE}(M(X_0))\neq z_t|X_0].\label{ProbInco}
\end{align}

\end{itemize}


\subsubsection{The single-query case} 
Combining \eqref{Eq40-Del} and Theorem \ref{T1}, we can characterize the achievable privacy guarantee in the single-query case as follows:
\begin{corollary}
Any $\epsilon_C$-DP single-query  ${M}$ under \eqref{outcome} is  $(\epsilon(t),t)$-age-dependent DP, where $\epsilon(t)$ satisfies
\begin{align}
   \epsilon(t)= \ln\left(1+|1-p-q|^t(\exp(\epsilon_C)-1)\right), \forall t\in\mathbb{N}.
\end{align}
\end{corollary}

		\begin{figure*}[t]
		\centering
		\subfigure[]{\includegraphics[scale=.27]{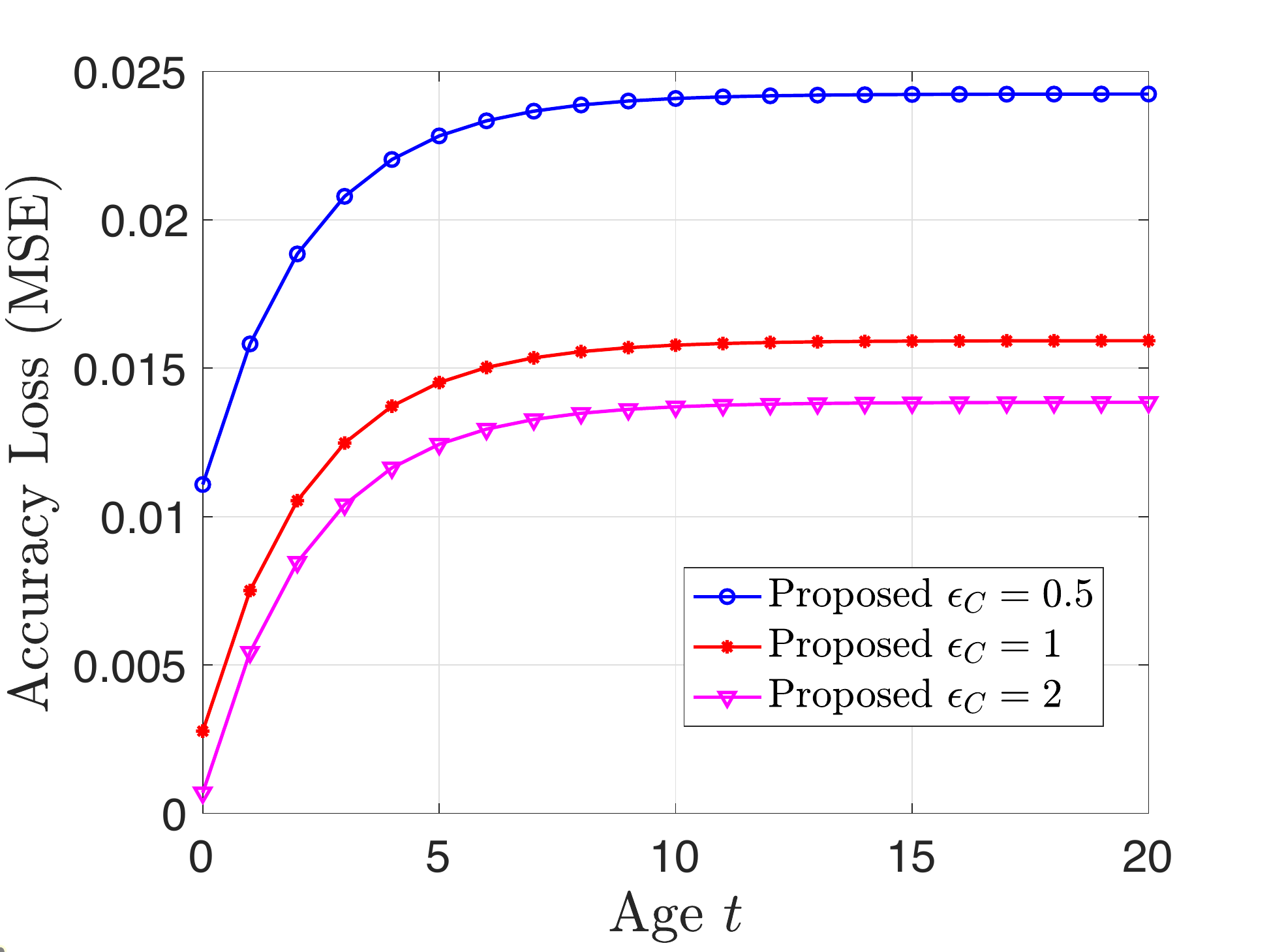}}
		\subfigure[]{\includegraphics[scale=.27]{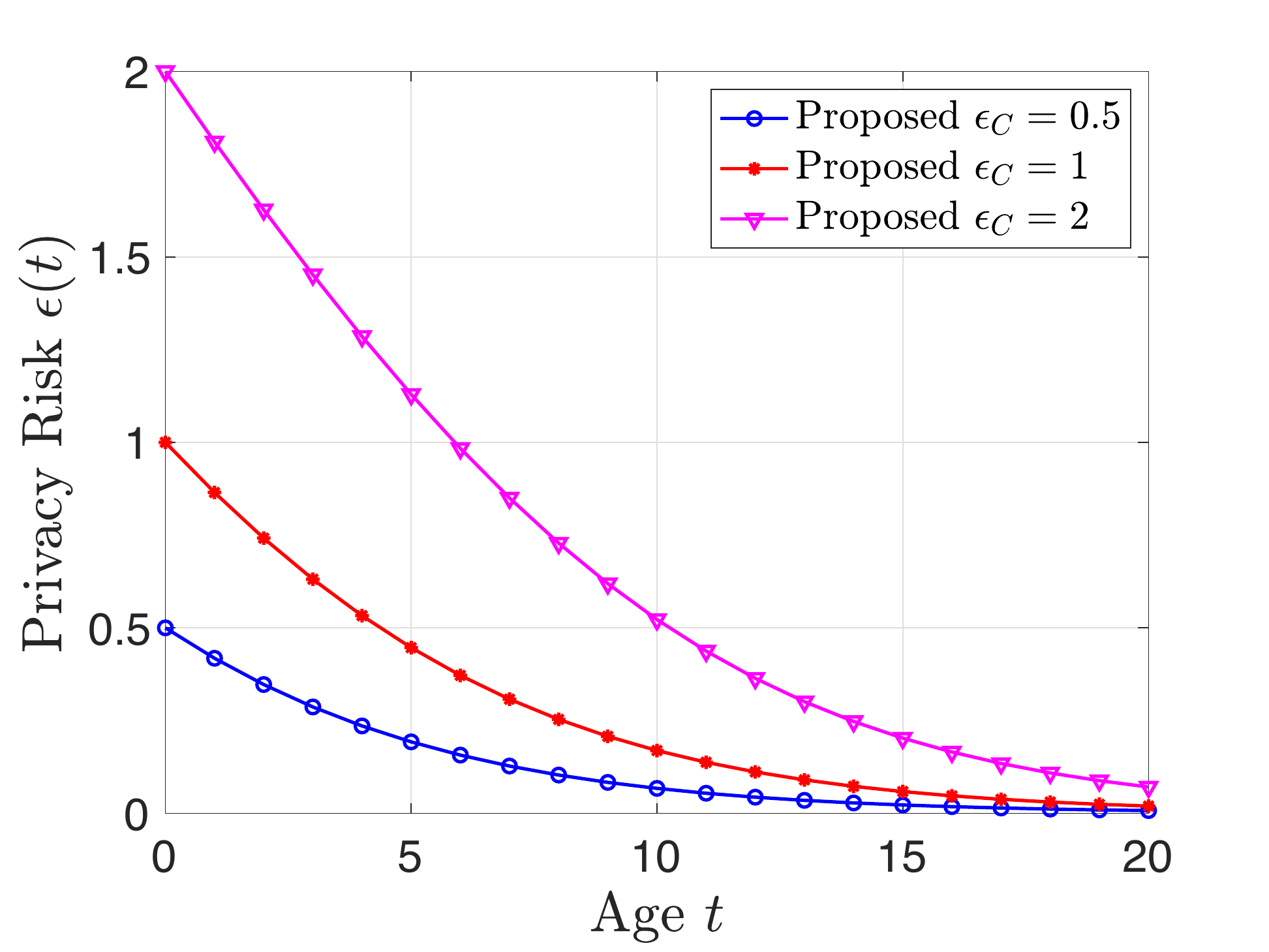}}
		\subfigure[]{\includegraphics[scale=.27]{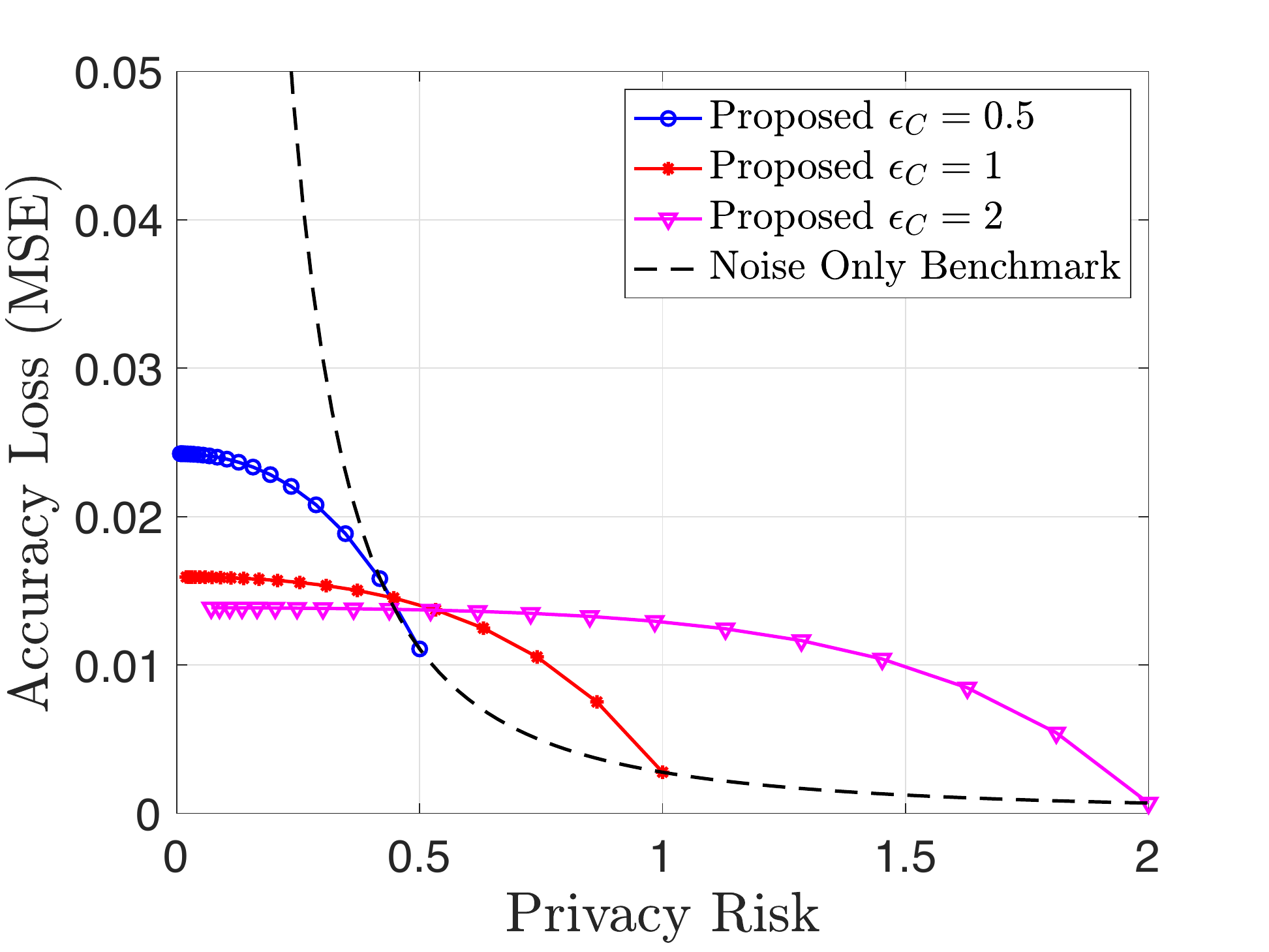}}
		\vspace{-5pt}
		\caption{Numerical results for Example \ref{eg1} with mean-square errors for single-query mechanisms. We set $p=q=0.1$ and $I=20$.
		} \label{figeg1}

		\vspace{-10pt}
	\end{figure*}

\subsubsection{The multi-query case} 
We further consider a multi-query mechanism with an equal-spacing policy described in \eqref{SE}.
The privacy guarantee result follows from Proposition \ref{P6}:
\begin{corollary}
Any $\epsilon_C$-DP multi-query $\mathcal{M}_t$ based on \eqref{outcome} is  $(\epsilon(t),{t},\mathcal{S},\mathcal{A})$-age-dependent DP with a peak privacy risk given by 
\begin{align}
   \epsilon^*= \ln \left(1+\frac{|1-p-q|^{A}\cdot(\exp(\epsilon_C)-1)}{1-|1-p-q|^{\bar{S}}\exp(\epsilon_C)}\right).
\end{align}
\end{corollary}

\subsubsection{Numerical Results} 
To understand the tradeoffs between privacy and utility for Example \ref{eg1}, we present numerical results for the mean-squared error and the failure rate scenarios in Figs. \ref{figeg1} and \ref{figeg2}, respectively. 

In the mean-square error scenario,
we present the accuracy loss and privacy risk at different ages $t$ in Figs. \ref{figeg1}(a) and \ref{figeg1}(b), respectively, and the tradeoffs between privacy risk and accuracy loss in Figure \ref{figeg1}(c). 
In Figure \ref{figeg1}(a), we show that the accuracy loss increases in $t$, and for a given $\epsilon_C$, it converges to an upper bound as $t\rightarrow \infty$. This implies that aging does not necessarily incur significant accuracy loss in terms of the mean-square error.
On the other hand, Figure \ref{figeg1}(b) shows that privacy risks decrease in $t$ and converge to $0$, which is mainly because $\Delta(t)$ converges to $0$.
In Figure \ref{figeg1}(c), we compare our proposed mechanisms against a classical noise only benchmark that only injects Laplace noise, whose accuracy loss is given by $2/(\epsilon_C^2 I^2)$. We show that, 
to achieve an arbitrary small privacy risk, the accuracy loss incurred by the benchmark grows unbounded. The large loss  is partially because of the lack of prior information in the considered mean-square error scenario.
Our proposed scheme combines both aging and noise injection and is able to achieve a finite accuracy loss, as the privacy risk approaches $0$. We note that this is partially because our age-dependent DP protects privacy specifically from adversaries seeking to infer $X_t$, whereas DP does not assume any type of data distributions or adversaries.

		\begin{figure*}[t]
		\centering
		\subfigure[]{\includegraphics[scale=.27]{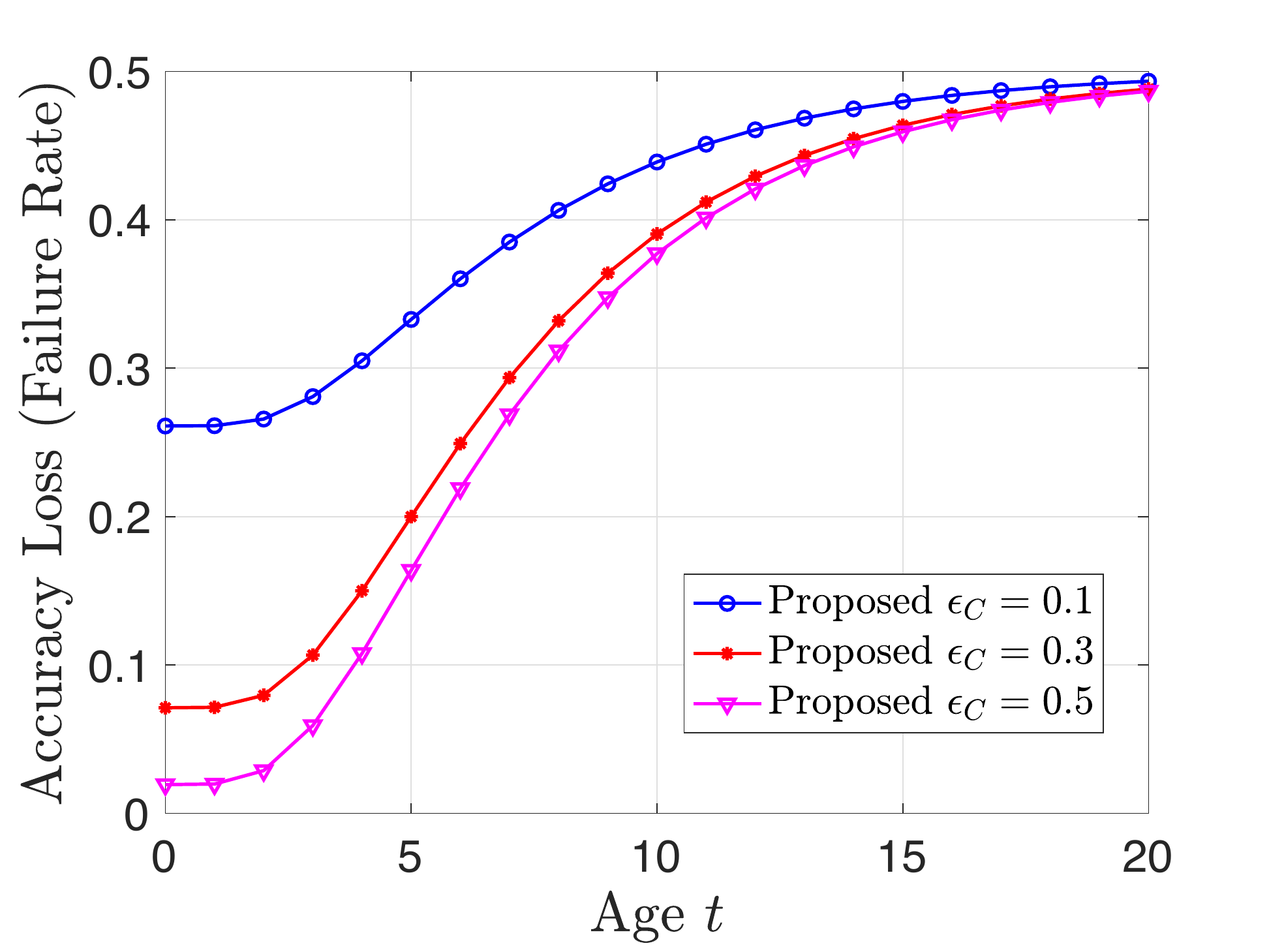}}
		\subfigure[]{\includegraphics[scale=.27]{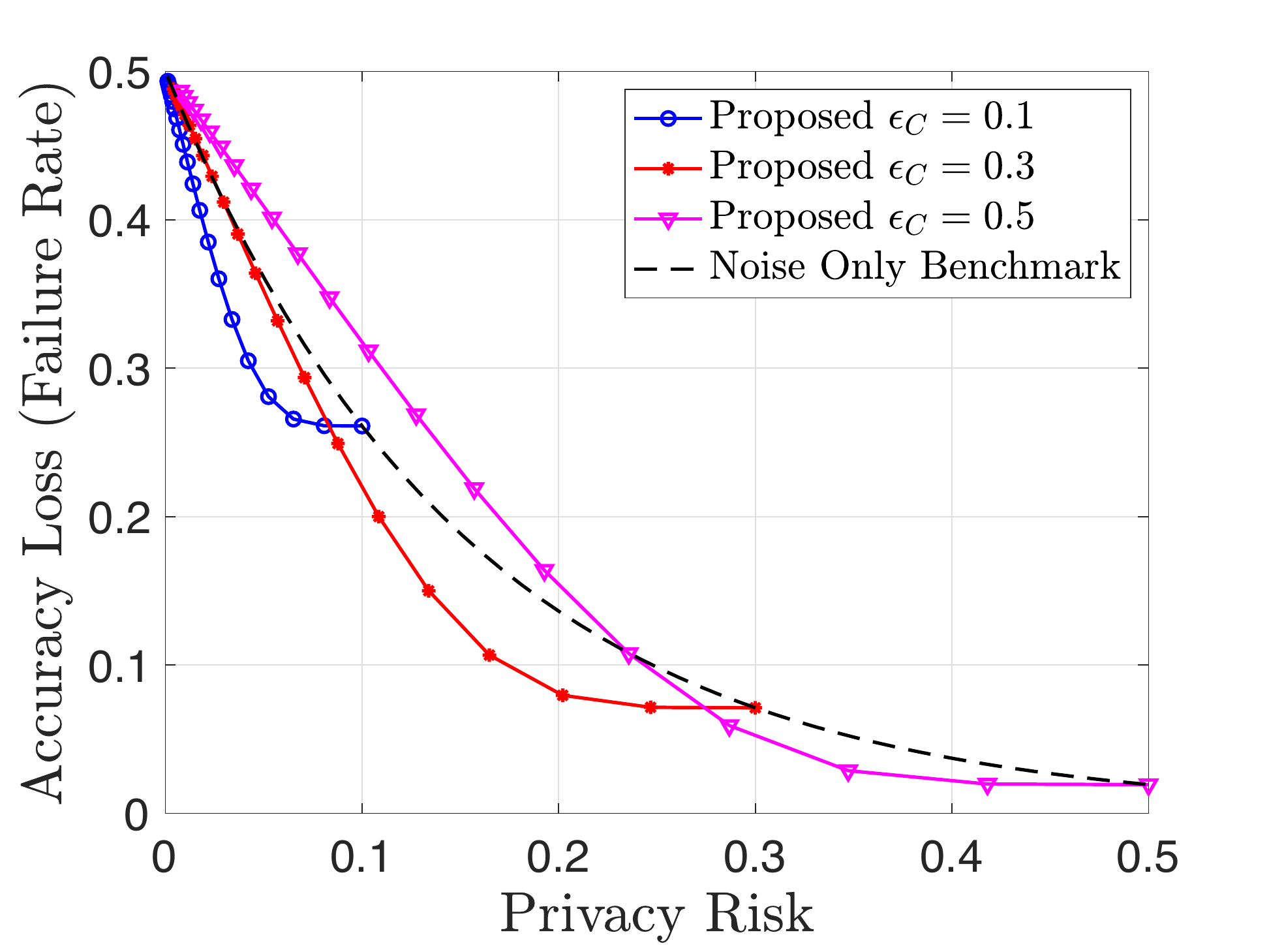}}
		\subfigure[]{\includegraphics[scale=.27]{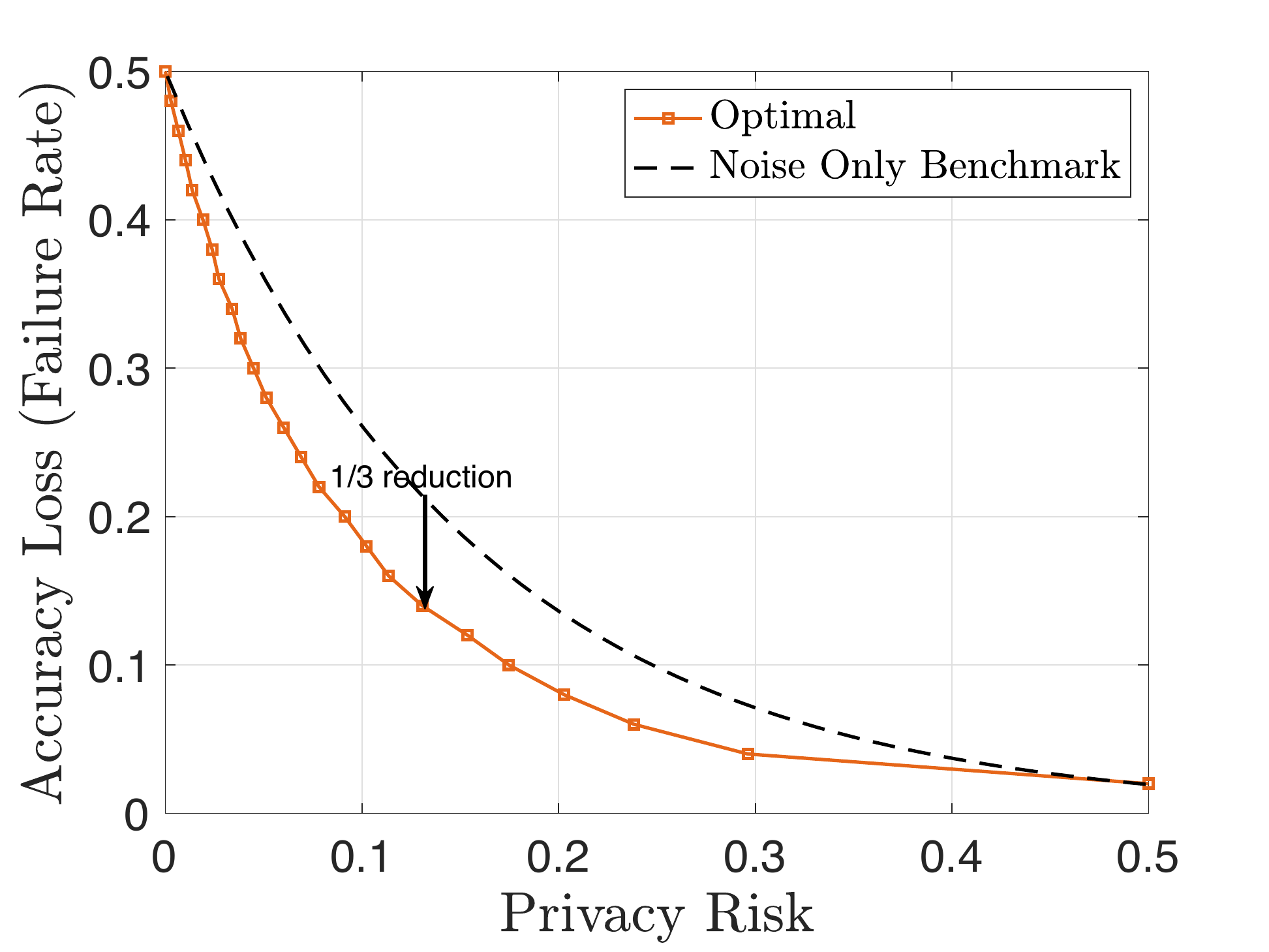}}
		\vspace{-5pt}
		\caption{Numerical results for Example \ref{eg1} with the failure rate  for single-query mechanisms. Here $p=q=0.1$ and $I=20$.
		} \label{figeg2}

		\vspace{-10pt}
	\end{figure*}

In the failure rate scenario, we present the accuracy loss and privacy risk for a single query mechanism at different ages $t$ in Figure \ref{figeg2}(a), and the tradeoffs between privacy risk and accuracy loss in Figures \ref{figeg2}(b)(c). 
Figure \ref{figeg2}(a) shows that the accuracy loss increases in $t$ as it does in Figure \ref{figeg1}(a), but at a different rate. The accuracy loss eventually converges to $0.5$ as $t\rightarrow \infty$. We note that the privacy risks for a given $t$ and $\epsilon_C$ are the same as in the mean-squared error scenario  (Figure \ref{figeg1}(b)). In terms of the tradeoffs between privacy and accuracy loss, we show that our proposed scheme can achieve less accuracy loss compared to the noise only benchmark. In addition, Figure \ref{figeg2}(b) also shows that different utility-privacy pairs require different $t$ and $\epsilon_C$. In Figure \ref{figeg2}(c), we present the optimal tradeoffs between privacy and accuracy loss. Compared to the noise only benchmark, this shows that our proposed scheme can achieve a reduction up to $1/3$ in accuracy loss given  the same privacy risk.

		\begin{figure*}[t]
		\centering
		\subfigure[]{\includegraphics[scale=.27]{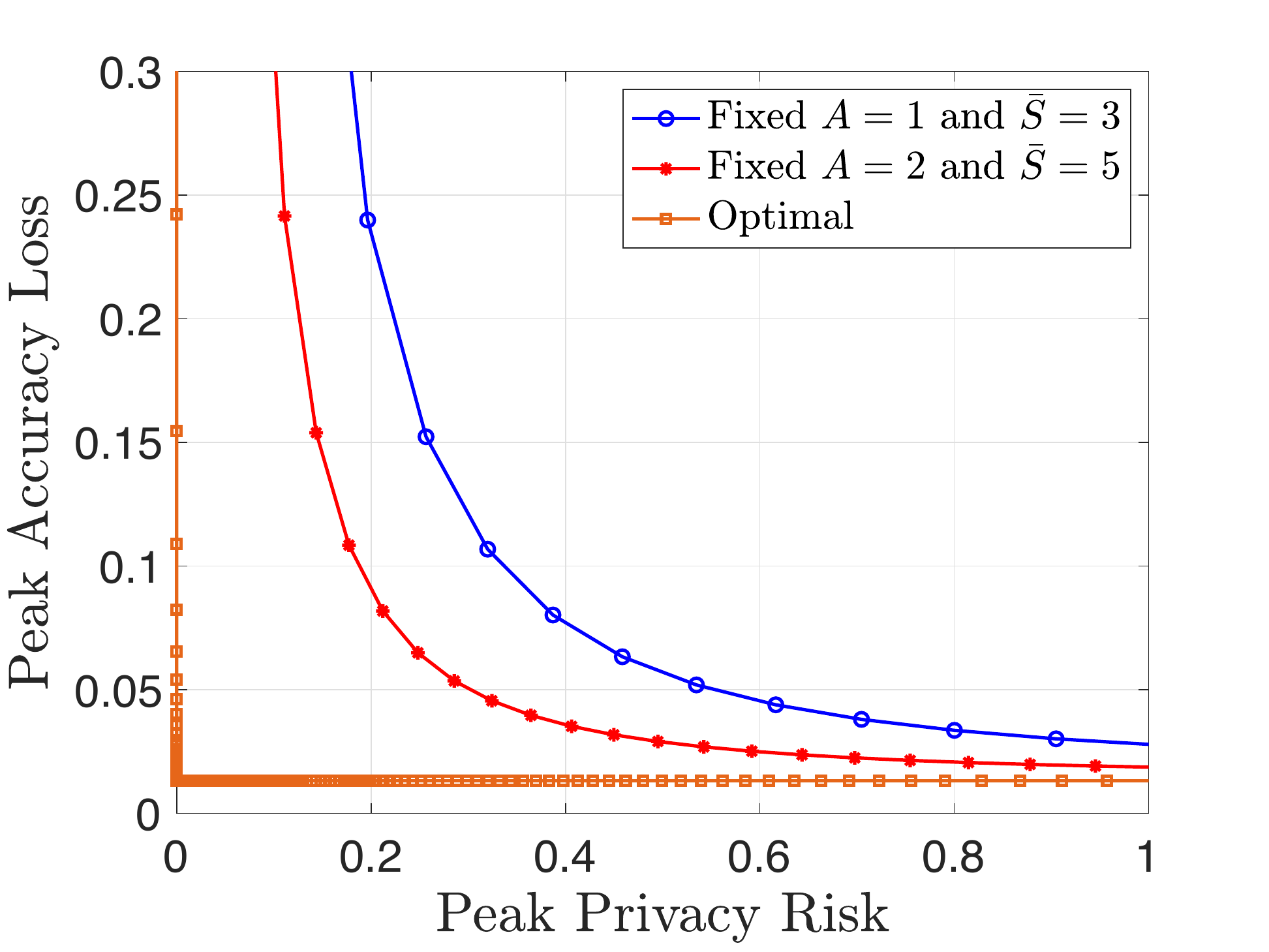}}
		\subfigure[]{\includegraphics[scale=.27]{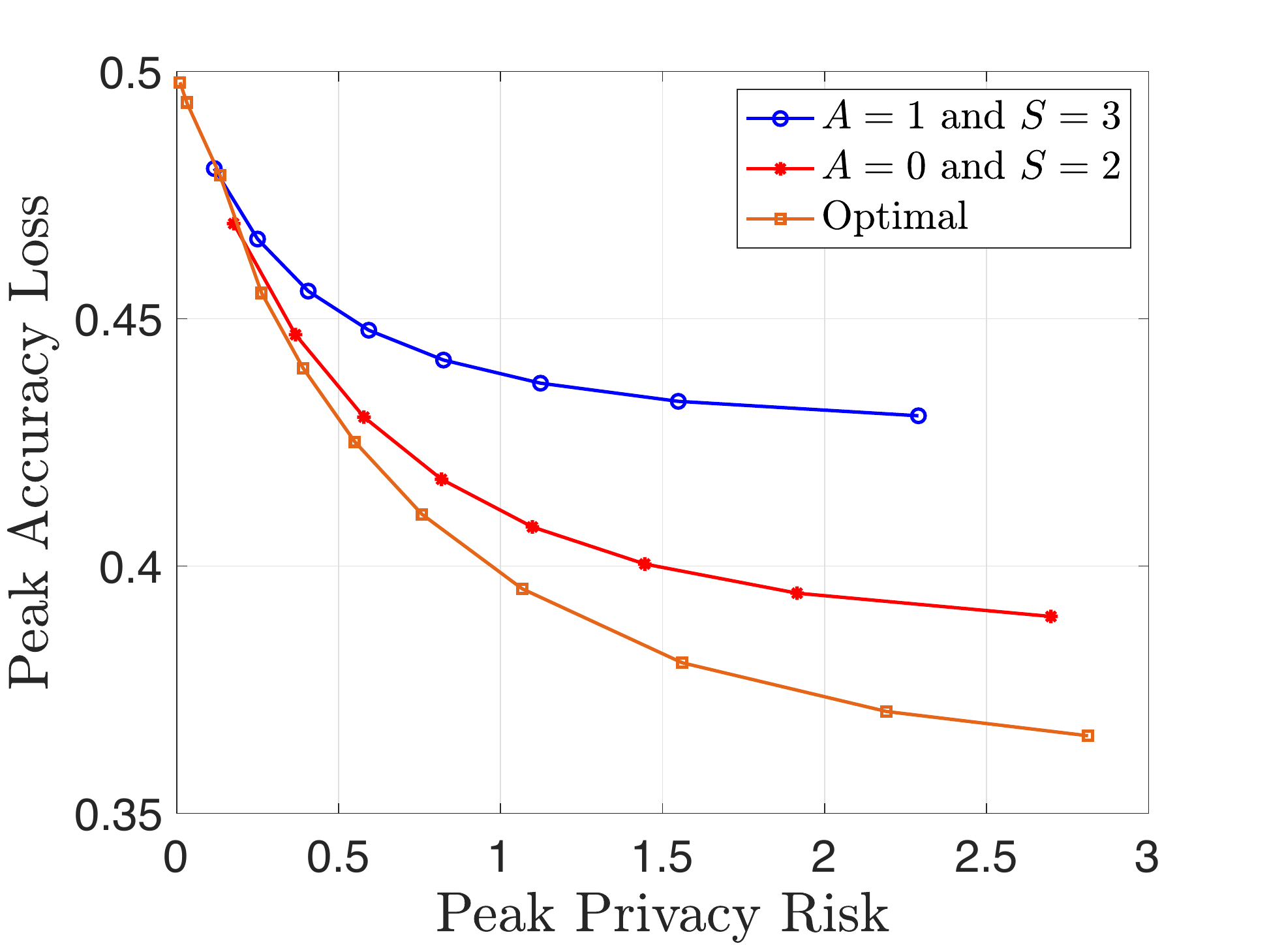}}
		\vspace{-5pt}
		\caption{Numerical results for Example \ref{eg1} with for  multi-query mechanisms with (a) mean-square errors and (b) the failure rate. We set $p=q=0.1$ and $I=20$.
		} \label{figeg34}

		\vspace{-10pt}
	\end{figure*}
	
Finally, Figure \ref{figeg34} presents the tradeoffs between privacy and utility in the multi-query scenario. 
Figure \ref{figeg34}(a) show that there exists a pair that minimizes both the peak accuracy loss and the peak privacy risk. This is mainly because the condition in Corollary \ref{Coro2} holds in this case and hence the tradeoff between privacy and utility does not exist. 
Figure \ref{figeg34}(b) demonstrates the existence of the optimal tradeoffs between privacy and utility for multi-query mechanisms in the failure rate scenario, as the condition in Corollary \ref{Coro2} does not hold in this case. That is, a smaller achievable peak privacy risk leads to a larger achievable peak accuracy loss.

\subsection{Electricity Consumption Forecast}

In this experiment, we use electricity consumption readings of households in London provided by \cite{London} to evaluate our proposed scheme. UK Power Networks recorded Power consumption (in kWh) every 30 minutes between November 2011 and February 2014. 
We selected 40 households, with $28,000$ readings per household on average. We
quantized the power values into 12 intervals, resulting in a Markov chain with $12$ states for each household. Related studies also considered such a Markov chain formulation of electricity consumption readings \cite{Pufferfish2}.
Our goal is to publish a privacy-preserving approximation to forecast households' average electricity consumption. In particular, to evaluate the privacy risks,
we use the following result to estimate an upper bound for the total variation distance between the transition probability and the stationary probability:
 \begin{proposition}[Bounds on variation distance \cite{Geo}]\label{P6}
 Let $\{X_t\}_{t\in\mathbb{N}}$ be an irreducible, aperiodic, and reversible Markov chain on a finite set $\mathcal{X}$. Then for all $X\in\mathcal{X}$, $t\in\mathbb{N}$, 
 \begin{align}
      4\cdot \delta (P_t(X,\cdot),\pi)^2\leq \frac{1-\pi(X)}{\pi(X)}\lambda_*^{2t}.
 \end{align}
 \end{proposition}
Based on Proposition \ref{P6}, we can bound the maximal total variation distance $\Delta(t)$, i.e., 
\begin{align}
    \Delta(t)\leq \min\left\{1,\max_{i\in\mathcal{I}}\max_{x_i\in\mathcal{X}_i}\sqrt{\frac{1-\pi_i(x_{i})}{\pi_i(x_{i})}}\lambda_{i,*}^t\right\}, \label{EC-D}
\end{align}
where $1-\lambda_{i,*}$ represents the spectral gap of household $i$'s time-varying database $\{x_{i,t}\}_{t\in\mathbb{N}}$.

		\begin{figure*}[t]
		\centering
		\subfigure[]{\includegraphics[scale=.27]{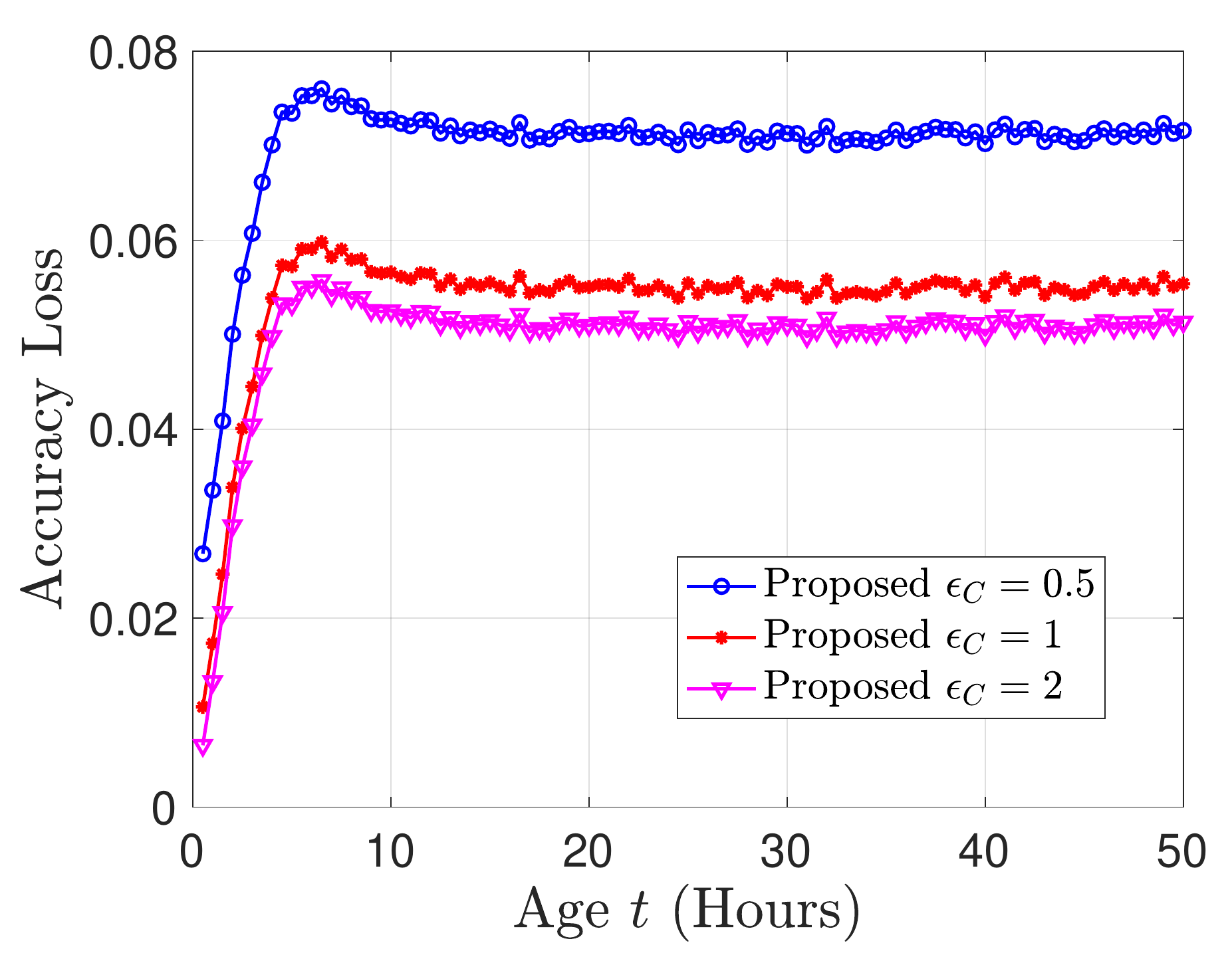}}
		\subfigure[]{\includegraphics[scale=.27]{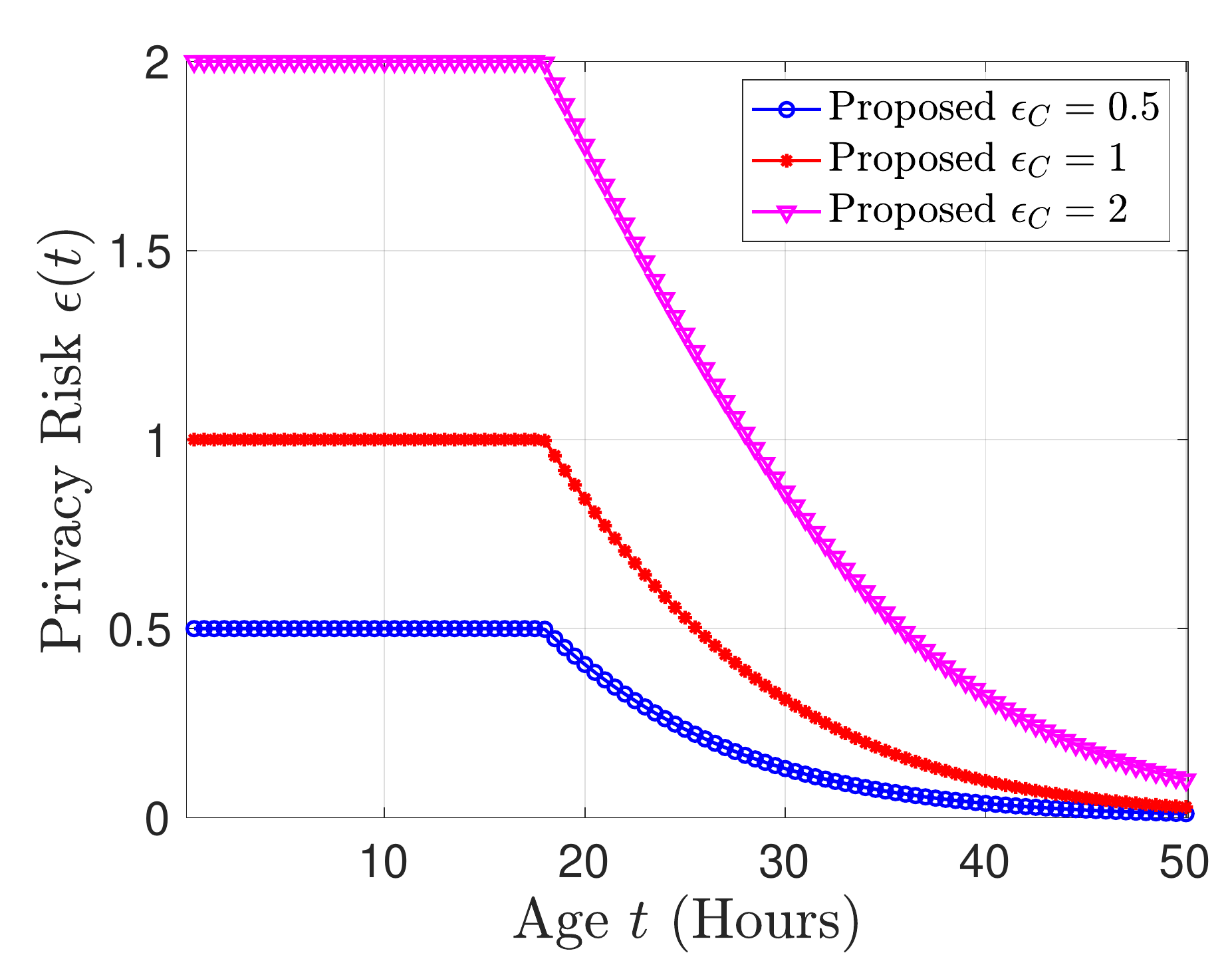}}
		\subfigure[]{\includegraphics[scale=.27]{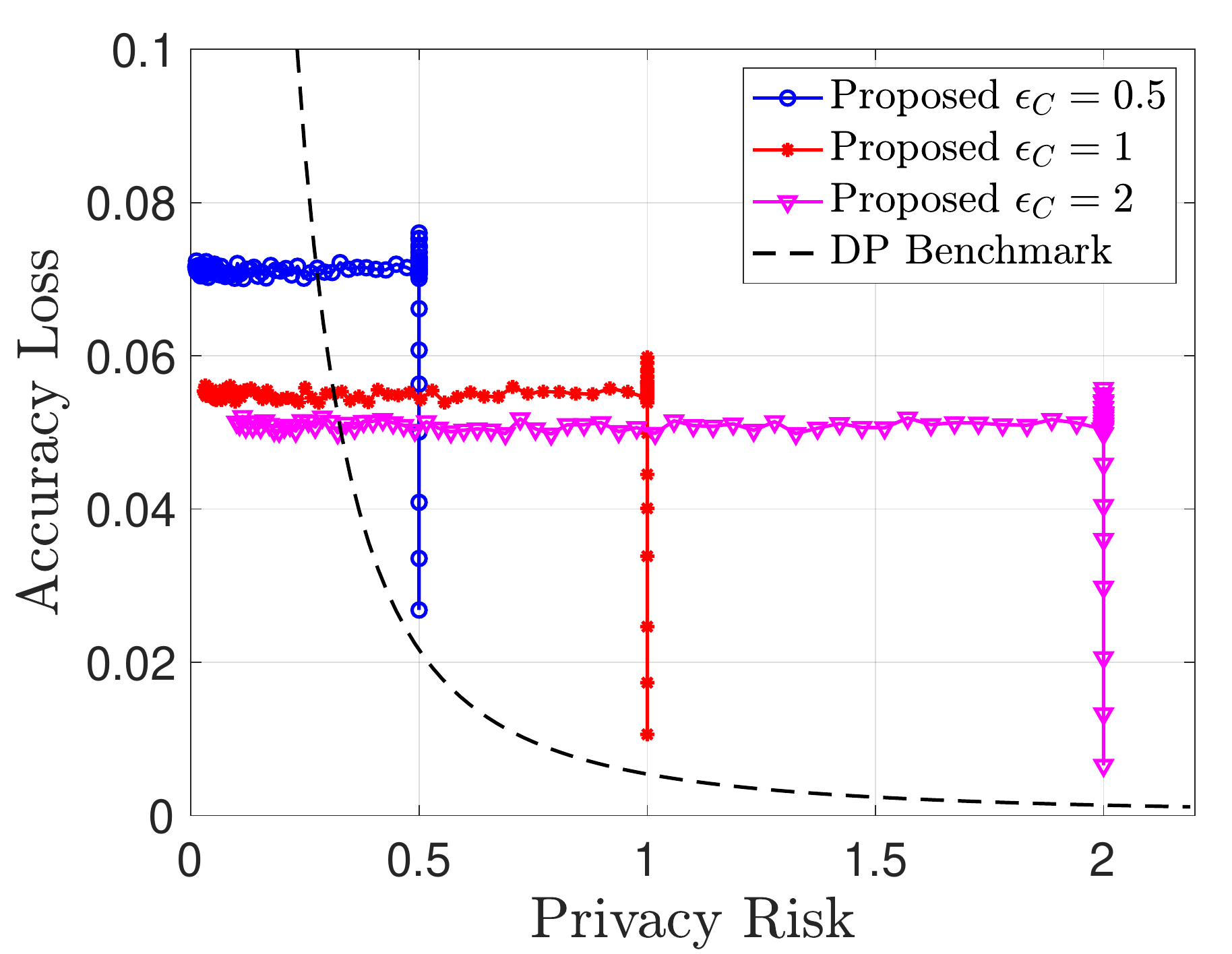}}
		\vspace{-5pt}
		\caption{Experimental results for electricity consumption data of London households.
		} \label{fig3}

		\vspace{-10pt}
	\end{figure*}
	
In Figure \ref{fig3}, we study the accuracy loss (measured by the MSE between the released data mean usage and the true mean usage) and the privacy risks
of our single-query mechanism applied to the electricity power dataset for different $\epsilon_C$ values.
As shown in Figure \ref{fig3}(a), the accuracy loss first increases rapidly in $t$ and then remains a relatively constant level, for all choices of $\epsilon_C$ values. In Figure \ref{fig3}(b),  we observe that the upper bound on the privacy risks in \eqref{EC-D} are first the same as their corresponding classical DP risks and then decrease when the age is greater than $18$ hours. This mainly results  from \eqref{EC-D}, in which the second term $\max_{i\in\mathcal{I}}\max_{x_i\in\mathcal{X}_i}\sqrt{\frac{1-\pi_i(x_{i})}{\pi_i(x_{i})}}$ only becomes less than $1$ when $t\geq 18$ hours.\footnote{It is possible to derive a tighter bound than \eqref{EC-D}, which is, however, out of the scope of this paper.} Finally, Figure \ref{fig3}(c) depicts the achievable tradeoffs between privacy and accuracy loss. Similarly, the accuracy loss incurred by the DP benchmark grows unbounded as the privacy risk approaches zero, while a finite accuracy loss is always achievable under our proposed scheme that combines both noise injection and aging.


	
\section{Conclusion} \label{Sec:Conclusion}

In this paper, we proposed an age-dependent generalization of differential privacy. We characterized the impact of data staleness on data privacy guarantees and showed that aging provides  a new direction to protect data privacy (in addition to noise injection) based on stale data. 
Motivated by frequent data updates in real-time applications, 
we further characterized how age-dependent privacy risks compose, given any publishing and aging policies. Finally, our case studies showed that mechanisms combining aging and noise injection may significantly outperform the classical mechanisms (that inject noise only), when making tradeoffs between privacy and utility.

As a first study on understanding the impact of data freshness on privacy, there are many future research directions. First, our work assumes that the adversary is only interested in keeping track on the most current state of some user. One potential direction is to extend our results to different adversary models, e.g., adversaries that aim  to infer a subset of the entire history of users' databases. Second, it is interesting to study adaptive composition to enhance the privacy guarantees for multi-query mechanisms.



\bibliographystyle{ACM-Reference-Format}

\appendix

\section{Appendix}

\subsection{Basic Properties}\label{Basic}
\begin{proposition}[Post-processing]
Let $M: \mathcal{X}\rightarrow \mathcal{Y}$ be a single-query  $(\epsilon,t)$-age-dependent differentially private  mechanism. Then, for any  mapping $f:\mathcal{Y}\rightarrow \mathcal{Y}'$, $f(M(\mathcal{X}_0))$ is $(\epsilon,t)$-age-dependent differentially private.
 
\end{proposition}

\begin{proof}
By the definition of age-dependent DP in Definition \ref{ADP}, we have:
\begin{align}
       {\rm Pr}[{M}(X_0)\in\mathcal{W}|X_t=\bs{x}]\leq \exp(\epsilon)  {\rm Pr}[{M}(X_0)\in\mathcal{W}|X_t=\bs{x}'],
\end{align}
for any set of outcomes $\mathcal{W}\subset \mathcal{Y}$.

For a mapping $f(\cdot)$, it follows that
\begin{align}
     {\rm Pr}[f({M}(X_0))\in\mathcal{T}|X_t=\bs{x}]&= {\rm Pr}[{M}(X_0)\in\mathcal{W}|X_t=\bs{x}]\nonumber\\
     &\leq \exp(\epsilon){\rm Pr}[{M}(X_0)\in\mathcal{W}|X_t=\bs{x}']\nonumber\\
     &=
     \exp(\epsilon){\rm Pr}[f({M}(X_0))\in\mathcal{T}|X_t=\bs{x}']
\end{align}
where $\mathcal{W} = f^{-1}[\mathcal{T}]$, i.e. $\mathcal{S}$ is the preimage of $\mathcal{T}$ under mapping $f$.
\end{proof}

\subsection{Proof of Lemma \ref{claim2}}\label{Proof-L1}
It follows from the definition of $\epsilon_C$-DP that 
\begin{align}
    \hat{{P}}_{-i,t}(\bs{x}_{-i},\bs{z}_{-i}){\rm Pr}(\mathcal{M}(z_i,\bs{z}_{-i})\in\mathcal{W})\leq \exp(\epsilon_C) \hat{{P}}_{-i,t}(\bs{x}_{-i},\bs{z}_{-i}) {\rm Pr}(\mathcal{M}(z_i',\bs{z}_{-i})\in\mathcal{W})\label{Ze}
\end{align}
for all $\bs{x}_{-i}$, $\bs{z}_{-i}\in\mathcal{X}_{-i}\triangleq \prod_{j\neq i} \mathcal{X}_j$, and any pair of $z_i$ and $z_i'\in\mathcal{X}_i$. Summing \eqref{Ze} over all $\bs{z}_{-i}\in\mathcal{X}_{-i}$ yields Lemma \ref{claim2}.

\subsection{Proof of Theorem \ref{T1}}\label{ProofT2}

In Lemma \ref{Claim3}, we bound $\sum_{i\in\mathcal{X}_i} \left(p(x_i)+\delta(x_i)\right) x_i$ based on both the constraint on $\delta(x_i)$ and the bound of $x_i$. 
The key difference in this proof from that of Theorem \ref{T2} is that we 
consider a lemma which does not require the bound of $x_i$, based on which we can prove Theorem \ref{T1}.
We start with such a lemma:
\begin{lemma}\label{L6}
For any part of vectors with lengths $K$ denoted by $\{a_1,a_2, ..., a_k\}$ and $\{b_1,b_2, ..., b_k\}$, if
   $g\leq \frac{a_k}{b_k}\leq G,$
then 
$g\leq \frac{\sum_{k=1}^{K}a_k}{\sum_{k=1}^{K}b_k}\leq G.
$
\end{lemma}

\begin{proof}
We can prove Lemma \ref{L6} by summing
$b_k g\leq {a_k}\leq b_k G$ all over $1\leq k\leq K$, and the dividing the all sides of the resultant inequality by $\sum_{k=1}^{K}b_k$.
\end{proof}

    

Without loss of generality, we assume that $X$ and $X'$ only differs in the $i$-th user's data, i.e., we have $X=\{x_{i}\}_{i\in\mathcal{I}}$ and $X'=\{x_{1}, ..., x_{i-1}, x_i', x_i, ..., x_I\}$. It follows that
\begin{align}
\ln\left(\frac{{\rm Pr}[M(X_0)\in\mathcal{W}|X_t=X]}{{\rm Pr}[M(X_0)\in\mathcal{W}|X_t=X']}\right)= &\ln\left(\frac{\sum_{\bs{z}\in\mathcal{X}} {\rm Pr}[M(\bs{z})\in\mathcal{W}] {\rm Pr}[X_0=\bs{z}|X_t=X]}{\sum_{\bs{z}\in\mathcal{X}} {\rm Pr}[M(\bs{z})\in\mathcal{W}] {\rm Pr}[X_0=\bs{z}|X_t=X']}\right)\nonumber\\
= &\ln\left(\frac{\sum_{{z}_i\in\mathcal{X}_i} \hat{P}_{i,t}(x_i,z_i) \sum_{\bs{z}_{-i}\in\mathcal{X}_{-i}}  \prod_{j\neq i} \hat{P}_{j,t}(x_j,z_j) {\rm Pr}[M(\bs{z})\in\mathcal{W}] }{\sum_{{z}_i\in\mathcal{X}_i} \hat{P}_{i,t}(x_i',z_i) \sum_{\bs{z}_{-i}\in\mathcal{X}_{-i}}  \prod_{j\neq i} \hat{P}_{j,t}(x_j,z_j){\rm Pr}[M(\bs{z})\in\mathcal{W}]} \right)\nonumber\\
= &\ln\left(
\frac{\mathbb{E}_{z_i\sim \hat{P}_{i,t}(x_i,\cdot)}\{ \mathbb{E}_{\bs{z}_{-i}\sim \hat{P}_{i,t}(\bs{x}_{-i},\cdot)} [{\rm Pr}(M(\bs{z})\in\mathcal{W})]\}}{\mathbb{E}_{z_i\sim \hat{P}_{i,t}(x_i',\cdot)}\{ \mathbb{E}_{\bs{z}_{-i}\sim \hat{P}_{i,t}(\bs{x}_{-i},\cdot)} [{\rm Pr}(M(\bs{z})\in\mathcal{W})]\}}\right)\nonumber\\
\overset{(a)}{\leq} &\ln\left(\max_{x_i,x_i',z_i}  \frac{\hat{P}_{i,t}(x_i,z_i)}{\hat{P}_{i,t}(x'_i,z_i)}\right)\nonumber\\
\overset{(b)}{\leq} &\max_{{x}_i,{x}_i',{z}_i}\frac{P_{i,t}(z_i,x_i)-P_{i,t}(z_i,x_i')}{P_{i,t}(z_i,x_i')}\nonumber\\
\overset{(c)}{\leq} &\max_{z_i}\frac{\sqrt{\frac{1-\pi(z_i)}{\pi(z_i)}}\lambda_*^{t}}{\pi(z_i)-\sqrt{\frac{1-\pi(z_i)}{4\pi(z_i)}}\lambda_*^{t}}=\mathcal{O}(\lambda_*^t),\label{Z1}
\end{align}
where (a) is due to Claim 1 and the fact that $X$ and $X'$ only differs in the $i$-th user's data; (b) is because $\ln(1+x)\leq x$; (c) is from Proposition \ref{P1}.

This completes the proof of Theorem \ref{T1}.

\subsection{Proof of Theorem \ref{T3}}\label{ProofT3}

Let us first rewrite $\mathcal{M}_t(\mathcal{H}_t(\mathcal{S},\mathcal{A}))$ into
$[\mathcal{M}_{S_{n-1}}(\mathcal{H}_{S_{n-1}}(\mathcal{S},\mathcal{A}))~~M_n(X_{S_n-A_n})]$ for $t\in(S_{n-1}, S_n]$. For notational simplicity in this proof, we denote $\mathcal{M}_{S_{n-1}}$ by $\mathcal{M}'$ and $M_n$ by ${M}$, respectively. We further denote $\mathcal{H}_t(\mathcal{S},\mathcal{A})$ and $X_{S_n-A_n}$ by $\bs{z}^{(1)}$ and $\bs{z}^{(2)}$, respectively.

For 
\begin{align}
    & \ln\left(\frac{{\rm Pr}[\mathcal{M}'\in\mathcal{W}_t, {M}(X_{S_n-A_n})\in\mathcal{W}_{S_n-A_n}|X_{t}=X]}{{\rm Pr}[\mathcal{M}'\in\mathcal{W}_t,{M}(X_{S_n-A_n})\in\mathcal{W}_{S_n-A_n}|X_{t}=X']}\right)\nonumber\\
    = &\ln\left(\frac{\sum_{\bs{z}^{(1)}} \sum_{\bs{z}^{(2)}}{\rm Pr}[\mathcal{M}'(\bs{z}^{(1)})\in\mathcal{W}_1]{\rm Pr}[M(\bs{z}^{(2)})\in\mathcal{W}_2] {\rm Pr}[X_{S_{n-1}}=\bs{z}^{(1)},X_{S_n-A_n}=\bs{z}^{(2)}|X_{t}=X]}{\sum_{\bs{z}^{(1)}} \sum_{\bs{z}^{(2)}}{\rm Pr}[\mathcal{M}'(\bs{z}^{(1)})\in\mathcal{W}_1]{\rm Pr}[M(\bs{z}^{(2)})\in\mathcal{W}_2] {\rm Pr}[X_{S_{n-1}}=\bs{z}^{(1)},X_{S_n-A_n}=\bs{z}^{(2)}|X_{t}=
   X']}\right)\nonumber\\
     = &\ln\left(\frac{{\sum_{\bs{z}^{(2)}} {\rm Pr}[X_{S_n-A_n}=\bs{z}^{(2)}|X_{t}=\bs{x}] {\rm Pr}[M(\bs{z}^{(2)})\in\mathcal{W}_2] {\rm Pr}[\mathcal{M}'(X_{S_n-A_n})\in\mathcal{W}_1|X_{S_n-A_n}=\bs{z}^{(2)}]}}{\sum_{\bs{z}^{(2)}} {\rm Pr}[X_{S_n-A_n}=\bs{z}^{(2)}|X_{t}=\bs{x}'] {\rm Pr}[M(\bs{z}^{(2)})\in\mathcal{W}_2] {\rm Pr}[\mathcal{M}'(X_{S_n-A_n})\in\mathcal{W}_1|X_{S_n-A_n}=\bs{z}^{(2)}]}\right)\nonumber\\
       = &\ln\left(\frac{\mathbb{E}_{\bs{z}^{(2)}\sim \hat{P}_{t_2}(\bs{x},\cdot)}\left[{\rm Pr} ({M}(\bs{z}^{(2)})\in\mathcal{W}_{2} ) \cdot {\rm Pr}[\mathcal{M}'(X_{S_n-A_n})\in\mathcal{W}_1|X_{S_n-A_n}=\bs{z}^{(2)}]\right]}{\mathbb{E}_{\bs{z}^{(2)}\sim \hat{P}_{t_2}(\bs{x}',\cdot)}\left[{\rm Pr} ({M}(\bs{z}^{(2)})\in\mathcal{W}_{2} )\cdot {\rm Pr}[\mathcal{M}'(X_{S_n-A_n})\in\mathcal{W}_1|X_{S_n-A_n}=\bs{z}^{(2)}]\right]}\right)\nonumber\\
      = &\ln\left(\frac{\mathbb{E}_{{z}_i^{(2)}\sim \hat{P}_{i,t_2}({x}_i,\cdot)}\left\{\mathbb{E}_{\bs{z}_{-i}^{(2)}\sim \hat{P}_{-i,t_2}(\bs{x}_{-i},\cdot)}\left[{\rm Pr} (M(\bs{z}^{(2)})\in\mathcal{W}_{2} ) \cdot {\rm Pr}[\mathcal{M}'(X_{S_n-A_n})\in\mathcal{W}_1|X_{S_n-A_n}=\bs{z}^{(2)}]\right]\right\}}{\mathbb{E}_{{z}_i^{(2)}\sim \hat{P}_{i,t_2}({x}'_i,\cdot)}\left\{\mathbb{E}_{\bs{z}_{-i}^{(2)}\sim \hat{P}_{-i,t_2}(\bs{x}_{-i},\cdot)}\left[{\rm Pr} (M(\bs{z}^{(2)})\in\mathcal{W}_{2} ) \cdot {\rm Pr}[\mathcal{M}'(X_{S_n-A_n})\in\mathcal{W}_1|X_{S_n-A_n}=\bs{z}^{(2)}]\right]\right\}}\right)\nonumber\\
      = &\ln\left(\frac{\mathbb{E}_{{z}_i^{(2)}\sim \hat{P}_{i,t_2}({x}_i,\cdot)}\left\{\mathbb{E}_{\bs{z}_{-i}^{(2)}\sim \hat{P}_{-i,t_2}(\bs{x}_{-i},\cdot)}
      \left[h_1(\bs{z}^{(2)})\cdot h_2(\bs{z}^{(2)})\right]\right\}}{\mathbb{E}_{{z}_i^{(2)}\sim \hat{P}_{i,t_2}({x}'_i,\cdot)}\left\{\mathbb{E}_{\bs{z}_{-i}^{(2)}\sim \hat{P}_{-i,t_2}(\bs{x}_{-i},\cdot)}
      \left[h_1(\bs{z}^{(2)})\cdot h_2(\bs{z}^{(2)})\right]\right\}}\right)\nonumber\\
      \leq &  \ln\left(1+\Delta\left(t-S_{n}+A_{n}\right)\cdot \left(\exp(\epsilon_{C,n}+\epsilon(S_{n}-A_n))-1\right)\right),~\forall t\in[S_n, S_{n+1}), n\in\mathbb{N},\label{Eq1-Proof-T3}
\end{align}
where $t_2=t-S_n+A_n$, $h_1(X)={\rm Pr} (M(X)\in\mathcal{W}_{2})$ and $h_2={\rm Pr}[\mathcal{M}'(X_{S_n-A_n})\in\mathcal{W}_1|X_{S_n-A_n}=\bs{z}^{(2)}]$.
Note that
\begin{subequations}\label{Eq2-Proof-T3}
\begin{align}
    h_1(X)&\leq \exp(\epsilon_{C,n}) h_1(X'),\\
    h_2(X)&\leq \exp(\epsilon(S_n-A_n)) h_2(X')
\end{align}
\end{subequations}
for all $X,X'\in\mathcal{X}$ that differ only in one entry. Substituting \eqref{Eq2-Proof-T3} into the \eqref{Eq25} yields the last inequality of \eqref{Eq1-Proof-T3}. This completes the proof.



\subsection{Proof of Proposition \ref{P7}}\label{ProofP7}
Since $\ln(1+x)$ is monotonically increasing in $x$, we can rewrite \eqref{Optimization} as
 \begin{subequations}\label{Optimization-A}
 \begin{align}
    \min\quad &\frac{c\cdot\rho^{A}\cdot(\exp(\epsilon_C)-1)}{1-c^2\cdot\rho^{\bar{S}}\exp(\epsilon_C)}\\
    {\rm s.t.}~\quad & c^2\rho^{\bar{S}}\exp(\epsilon_C)<1,\label{Eq-Opt-Ab}\\
    & f(A+\bar{S},\epsilon_C)\leq \bar{f}, \label{Eq-Opt-Ac}\\
    {\rm var.}\quad & \epsilon_C\geq 0, A\geq 0, \bar{S}\geq 0,
\end{align}
 \end{subequations}
	We first drop the constraint in \eqref{Eq-Opt-Ab} and formulate the corresponding Lagrangian:
	\begin{align}
	    \mathcal{L}(A,\bar{S},\epsilon_C)= \frac{c\cdot\rho^{A}\cdot(\exp(\epsilon_C)-1)}{1-c^2\cdot\rho^{\bar{S}}\exp(\epsilon_C)}-\lambda (f(A+\bar{S},\epsilon_C)-\bar{f}),
	\end{align}
	where $\lambda$ is the dual variable corresponding to the constraint in \eqref{Eq-Opt-Ac}. The Karush–Kuhn–Tucker  conditions (necessary for all local optima) are \cite{Boyd}
	\begin{subequations}\label{KKT}
	\begin{align}
	  \frac{-c\cdot\rho^A \ln(\rho)(\exp(\epsilon_C)-1)}{c^2\rho^{\bar{S}}\exp(\epsilon_C)-1}-\lambda \frac{\partial f(A+\bar{S},\epsilon_C)}{\partial A}&=0,\\
	   \frac{c^3\cdot\rho^{A+\bar{S}}\ln(\rho)(\exp(\epsilon_C)-1)}{(c^2\rho^{\bar{S}}\exp(\epsilon_C)-1)^2}-\lambda \frac{\partial f(A+\bar{S},\epsilon_C)}{\partial \bar{S}}&=0.
	\end{align}
		\end{subequations}
	Combining \eqref{KKT} and the fact that $\frac{\partial f(A+\bar{S},\epsilon_C)}{\partial \bar{S}}=\frac{\partial f(A+\bar{S},\epsilon_C)}{\partial A}$, we can obtain \eqref{QS}.
	We note that i) \eqref{QS} is necessary even if the problem in \eqref{Optimization-A} is not convex;
	ii) \eqref{QS} automatically satisfies \eqref{Eq-Opt-Ab}, which implies that the relaxation of \eqref{Eq-Opt-Ab} does not lose any feasibility.
	

\subsection{Preliminaries: Mixing, Ergodicity, and Markov Chains}\label{Subsec:Pre}
    
  In order to characterize how privacy guarantees change over time under the new framework of age-dependent DP, we introduce several related concepts and preliminary results from the literature in this subsection.

\subsubsection{Geometric Ergodicity} 

To understand under what conditions a Markov chain satisfies the geometric ergodicity, we first introduce the following important result in analyzing the transition probability matrix $P(\cdot,\cdot)$:
\begin{lemma}[Bounds on Eigenvalues \cite{Geo}]\label{L1}
 Let $P(\cdot,\cdot)$ be the transition matrix of an irreducible Markov chain $\{X_t\}_{t\in\mathbb{N}}$. Then the eigenvalues of $P$, $\{\lambda_k\}$, satisfy:
\begin{align}
    1=\lambda_0>\lambda_1\geq ...\geq \lambda_{m-1}\geq -1,
\end{align}
where $m=|\mathcal{X}|$ denotes the total number of states.
\end{lemma}
 The largest eigenvalue $\lambda_0$ being $1$ results from irreducibility. It is known that the convergence rate largely depends on the value $\lambda_*=\max(\lambda_1,|\lambda_{m-1}|)$.

 There are a few studies in the literature of applied probability on analysis of the value of $\lambda_*$ (e.g., \cite{Geo,bound}). As an example, Landau and Odlyzko in \cite{bound} provided a bound for a random walk on a connected graph and showed the following result: 
\begin{proposition}[\cite{bound}]
Let $\{X_t\}_{t\in\mathbb{N}}$ be a random walk on a connected graph, it follows that
\begin{align}
    \lambda_*\leq 1-\frac{1}{|\mathcal{X}|d_*(1+\gamma_*)}<1-\frac{1}{|\mathcal{X}|^3},
\end{align}
where $|\mathcal{X}|$ is the number of vertices in the graph, $d_*$ is the maximum degree of the graph, and $\gamma_*$ is the diameter  of the graph.
\end{proposition}

\subsection{Proof of Theorem \ref{T4}}\label{ProofT4}
To prove Theorem \ref{T4}, we will first prove that the monotonicity of $\epsilon(S_n)$ 
in $n$ (Lemma \ref{L7}). We will then present the intermediate result of the existence of an optimal solution that have identical $(A_n^*,\bar{S}_n^*, \epsilon_{C,n}^*)$ after $n\geq m$ in Lemma \ref{L8}, based on which we complete the whole proof.

Define
\begin{align} 
F(A_n,\bar{S}_n,\epsilon_{C,n},\epsilon)\triangleq \ln\left(1+\Delta\left(\bar{S}_n+A_{n}\right)\cdot \left(\exp(\epsilon_{C,n})[\Delta(A)\cdot(\epsilon-1)+1]-1\right)\right),
\end{align}
for all $t\in[S_n, S_{n+1}), n\in\mathbb{N}$.
We have $\epsilon_{n+1}=F(A_n,\bar{S}_n,\epsilon_{C,n},\epsilon_n)$, which is strictly increasing in $\epsilon$ and $\epsilon_{C,n}$.
For any $m\in\mathbb{N}$, if we fix ${A}_n,\bar{S}_n, \epsilon_{C,n}$ for all $n \geq m$, it follows that ${\epsilon_m+1}$ strictly increases in $\epsilon_m$.

Let the maximal objective value of \eqref{GeneralProblem} be $\epsilon^*$.
Let $m\in\mathbb{N}\cup\{+\infty\}$ be the first epoch that the peak age-dependent privacy risk, i.e., $\epsilon(S_m)=\epsilon^*$.


We first introduce the following lemma:
\begin{lemma}\label{L7}
For any optimal solution $\mathcal{Z}^*=(\{A_n^*\},\{\bar{S}_n^*\},\{\epsilon_{C,n}^*\})$, the sequence $\{\epsilon_n\}_{n}$ is non-decreasing in $n$.
\end{lemma}
To prove Lemma \ref{L7}, suppose that there exists an optimal solution $\mathcal{Z}^*=(\{A_n^*\},\{\bar{S}_n^*\},\{\epsilon_{C,n}^*\})$ such that
$\epsilon_{k+1}<\epsilon_{k}$ for some $k$. In this case, we replace $(A^*_{n},\bar{S}^*_{n},\epsilon^*_{C,n})$ for all $n\geq k$ by $(A^*_{k},\bar{S}^*_{k},\epsilon^*_{C,k})$, i.e., we construct the following new solution $\mathcal{Z}'=(\{A_n'\},\{\bar{S}_n'\},\{\epsilon_{C,n}'\})$ satisfying:
\begin{align}
   (A'_{n},\bar{S}'_{n},\epsilon'_{C,n})=\begin{cases}
    (A^*_{k},\bar{S}^*_{k},\epsilon^*_{C,k}),\quad{\rm if}\quad n\geq k,\\
    (A^*_{n},\bar{S}^*_{n},\epsilon^*_{C,n}),\quad{\rm otherwise}.
    \end{cases}
\end{align}
For the new sequence $\{\epsilon_n'\}$ generated by $\mathcal{Z}'$, it follows that
\begin{align}
    F(A^*_{k},\bar{S}^*_{k},\epsilon^*_{C,k},\epsilon_n')<F(A^*_{k},\bar{S}^*_{k},\epsilon^*_{C,k},\epsilon_n),
\end{align}
i.e., $\epsilon_{n+1}'<\epsilon_{n+1}$. Similarly, we have
\begin{align}
  \epsilon_{n+1}'<  \epsilon_n', \forall n\geq k.
\end{align}
Therefore, we have $\epsilon^*>\epsilon_{k}>\epsilon_{k}'>\epsilon_n$ for all $n\in\mathbb{N}$, 
 which is a contradiction to the fact that $\mathcal{Z}^*$ is optimal.

We have the new solution $\mathcal{Z}'$ is also feasible (as $f(A_m+\bar{S}_m,\epsilon_{C,m})\leq \bar{f}$). 
\begin{align}
   F(A_k^*,\bar{S}_k^*,\epsilon_{C,n}^*)= \ln\left(1+\Delta\left(\bar{S}_k^*+A_{k}^*\right)\cdot \left(\exp(\epsilon_{C,n}^*)[\Delta(A_k^*)\cdot(\epsilon_n-1)+1]-1\right)\right), n\in\mathbb{N}.
\end{align}

In addition, we must have
\begin{lemma}\label{L8}
For any optimal solution $\mathcal{Z}^*=(\{A_n^*\},\{\bar{S}_n^*\},\{\epsilon_{C,n}^*\})$, the sequence $\{\epsilon_n\}_{n}$, let $\epsilon_m=\epsilon_n^*$ for $m$, then there exists another new optimal solution such that $A_n^*=A^*$, $S_n^*=S^*$, and $\epsilon_{C,n}^*=\epsilon_{C}^*$ for all $n\geq m$.
\end{lemma}
For any optimal solution $\mathcal{Z}^*=(\{A_n^*\},\{\bar{S}_n^*\},\{\epsilon_{C,n}^*\})$, we construct a new solution such that
\begin{align}
   (A'_{n},\bar{S}'_{n},\epsilon'_{C,n})=\begin{cases}
    (A^*_{m},\bar{S}^*_{m},\epsilon^*_{C,m}),\quad{\rm if}\quad n\geq m,\\
    (A^*_{n},\bar{S}^*_{n},\epsilon^*_{C,n}),\quad{\rm otherwise}.
    \end{cases}
\end{align}
It follows that
\begin{align}
   F(A^*_{m},\bar{S}^*_{m},\epsilon^*_{C,m},\epsilon_m)=\epsilon_{m+1}=\epsilon^*.
\end{align}



Now, we focus on the optimal solution $\bs{Z}^*$ such that $A_n^*=A_{n+1}^*=A^*$, and $\bar{S}_n^*=\bar{S}_{n+1}^*\bar{S}^*$, $\epsilon_{C,n}^*=\epsilon_{C,n+1}^*=\epsilon_{C}^*$ for all $n\geq m$. Now, we consider a new solution such that $\mathcal{Z}'=(\{A_n'\},\{\bar{S}_n'\},\{\epsilon_{C,n}'\})$ such that $A_n'=A^*$, $\bar{S}_n'=\bar{S}^*$, and $\epsilon_{C,n}'=\epsilon_{C}^*$  for all all $n\in\mathbb{N}$.

Note that, since 
\begin{align}
    F(A^*,\bar{S}^*,\epsilon_{C}^*,\epsilon^*)=\epsilon^*,
\end{align}
and $F(A^*,\bar{S}^*,\epsilon_{C}^*,\epsilon)$ in strictly increasing in $\epsilon$. We have that, for the sequence $\{\epsilon_n\}$ generated according to $\epsilon_{n+1}=F(A^*,\bar{S}^*,\epsilon_{C}^*,\epsilon_n)$ for all $n\in\mathbb{N}$ with $\epsilon_0=0$, $\epsilon_{n}\leq \epsilon^*$ for all $n\in\mathbb{N}$. Therefore, the new solution $\mathcal{Z}'$ must be optimal as well. 

\subsection{Proof of Corollary \ref{Coro2}}\label{Proof-Coro2}
Let $(A,\bar{S},\epsilon_C)$ be any feasible solution, we construct another solution $(A',\bar{S}',\epsilon_C')$ such that
\begin{align}
    A'>A,~~\bar{S}'> &~S,  \epsilon_C'>\epsilon_C,\\
    \rho^{A}(\exp(\epsilon_C)-1)= &~\rho^{A'}(\exp(\epsilon_C')-1),\\
    \rho^{\bar{S}}\exp(\epsilon_C)= &~\rho^{\bar{S}'}\exp(\epsilon_C').
\end{align}
From \eqref{Eq-Opt-a}, we can show that $(A',\bar{S}',\epsilon_C')$ leads to the same objective value as $(A,\bar{S},\epsilon_C)$ does.

On the other hand, 
we have that
\begin{align}
    A'&=A+\frac{\ln(\exp(\epsilon_C')-1)}{\ln(\rho^{-1})}-\frac{\ln(\exp(\epsilon_C)-1)}{\ln(\rho^{-1})},\\
    \bar{S}'&=\bar{S}+\frac{\epsilon_C'}{\ln(\rho^{-1})}-\frac{\epsilon_C}{\ln(\rho^{-1})}.
\end{align}
Therefore, the new value of the peak age penalty is given by
\begin{align}
    f\left(A'+\bar{S}',\epsilon_C'\right)=f\left(\frac{\ln(\exp(\epsilon_C')-1)}{\ln(\rho^{-1})}+\frac{\epsilon_C'}{\ln(\rho^{-1})}+a,\epsilon_C'\right),\label{peakage}
\end{align}
where $a=A-\frac{\ln(\exp(\epsilon_C)-1)}{\ln(\rho^{-1})}+\bar{S}-\frac{\epsilon_C}{\ln(\rho^{-1})}$.

The assumption that \eqref{peakage} is decreasing in $\epsilon'_C$  means that the new solution will satisfy  $f(A'+\bar{S}',\epsilon_C')<f(A+\bar{S},\epsilon_C)$ while achieving the same peak privacy guarantees.

\subsection{Autoregressive Model}\label{AGM}
We further consider a continuous-state 
autoregressive (AR) models, expressed as
\begin{align}
    x_{i,t+1}=\rho_1 x_{i,t}+\rho_2 x_{i,t-1}+\rho_3x_{i,t-2}... +\epsilon_{i,t}, ~~\forall i\in\mathcal{I}, t\in \mathbb{N}_+,\label{AR}
\end{align}
where $\epsilon_{i,t}$ follows an i.i.d. normal distribution $\mathcal{N}(0,\sigma^2)$. 

We start with the following time reversibility result:
\begin{lemma}[\cite{AR}]
When $\epsilon_{i,t}$ follows i.i.d. normal distributions, the process in \eqref{AR} is time reversible.
\end{lemma}
In the following, we focus on an AR(1) model for trackability, i.e., among all $\{\rho_{t}\}$ only the first coefficient $\rho_1$ takes a non-zero value. Therefore, we drop the index $1$ in $\rho_1$ in the following.

When $\epsilon_{i,t}$ has mean zero and variance $\sigma^2$, it follows that
\begin{align}
    \mathbb{E}[x_{i,t+n}|x_{i,n}]=\rho^t x_{i,n}, \forall t,n \in \mathbb{N}_+,
\end{align}
and 
\begin{align}
    \text{Var}[x_{i,t}|x_{i,0}]=\frac{(1-\rho^{2t}) \sigma^2}{1-\rho^2}, \forall i\in\mathcal{I}, \forall n \in\mathbb{N}_+.
\end{align}

\begin{lemma} When $\epsilon$ follows 
$\mathcal{N}(0,\sigma^2)$ for all $i\in\mathcal{I}$ and $t\in\mathbb{N}$, the probability distribution of $X_{i,t}$ conditional on $X_{i,0}$ is expressed as
\begin{align}
    x_{i,t}|x_{i,0}\sim \mathcal{N}\left(\rho^t x_{i,0},\frac{(1-\rho^{2t}) \sigma^2}{1-\rho^2}\right),~~\forall i\in\mathcal{I}, t\in \mathbb{N}_+.
\end{align}
\end{lemma}

In addition, we note that the total variation distance is closely related to another well-known metric called the \textit{Kullback–Leibler divergence}:
\begin{align}
    {\displaystyle D_{\text{KL}}(\mu\parallel \pi)\triangleq\sum _{x\in {\mathcal {X}}}\pi(x)\log \left({\frac {\mu(x)}{\pi(x)}}\right)},
\end{align}
and its  relation to the total variation distance can be expressed in the following:
\begin{lemma}\label{TVD}
The total variation distance is related to the Kullback–Leibler divergence by the Pinsker's inequality:
\begin{align}
\delta (\pi,\mu)\leq {\sqrt {{\frac {1}{2}}D_{\mathrm {KL} }(\pi\parallel \mu)}}, \forall \pi,\mu.
\end{align}
\end{lemma}

It follows that the Kullback-Leibler divergence is 
\begin{align}
D_{\mathrm {KL} }(P_t(x_{i,0},\cdot)\parallel P_t(x_{i,0}',\cdot))={1 \over 2}\left\{{\frac {(\mu _{1}-\mu _{0})^{2}}{\sigma ^{2}}}\right\}&=\frac{(1-\rho^2)\rho^{2t}(x_{i,t}-x'_{i,t})^2}{2\sigma^2 (1-\rho^{2n})}\nonumber\\
&\leq \frac{\rho^{2t}(x_{i,0}-x'_{i,0})^2}{2\sigma^2 }=\mathcal{O}(\rho^{2t}).
\end{align}
Based on the relation between the Kullback-Leibler divergence and the total variation distance in Lemma \ref{TVD}, it follows that
\begin{align}
   \Delta(t)= \delta (P_t(x_{i,0},\cdot),P_t(x_{i,0}',\cdot))\leq {\sqrt {{\frac {1}{2}}D_{\mathrm {KL} }(P_t(x_{i,0},\cdot)\parallel P_t(x_{i,0}',\cdot))}}\leq \frac{\rho^{t}|x_{i,0}-x'_{i,0}|}{2\sigma}.
\end{align}
The convergence rate of $\Delta(t)$ again is $\mathcal{O}(\rho^{t})$.

From Theorem \ref{T1}, we conclude that any $\epsilon_C$-DP mechanism ${M}$ is also $(\epsilon(t),t)$-age-dependent DP, satisfying
\begin{align}
     \epsilon(t)=\ln\left(1+\frac{\rho^{t}|x_{i,0}-x'_{i,0}|}{2\sigma }
     \cdot (\exp(\epsilon_C)-1)\right).
\end{align}

The single-query  mechanism ${M}^*$ that estimates the average value of $x_{i,t}$ over all users, i.e.,
\begin{align}
    {M}(X_0)=\left[\frac{\sum_{i}x_{i,0}}{I}\right]+n,
\end{align}
where $n$ is a Laplace $(1/\epsilon_C)$ random variable. Similarly, we use the mean-squared estimation error as the accuracy loss metric, given by
\begin{align}
         f_{\rm MMSE}(t,\epsilon_C)&= \mathbb{E}\left[\left(M(X_0)-\frac{1}{I}\sum_{i}x_{i,t}\right)^2\right], \\   &= \frac{(1-\rho^{2t}) \sigma^2}{I(1-\rho^2)} +\frac{1}{\epsilon_C^2I^2}.
\end{align}



		\begin{figure*}[t]
		\centering
		\subfigure[]{\includegraphics[scale=.27]{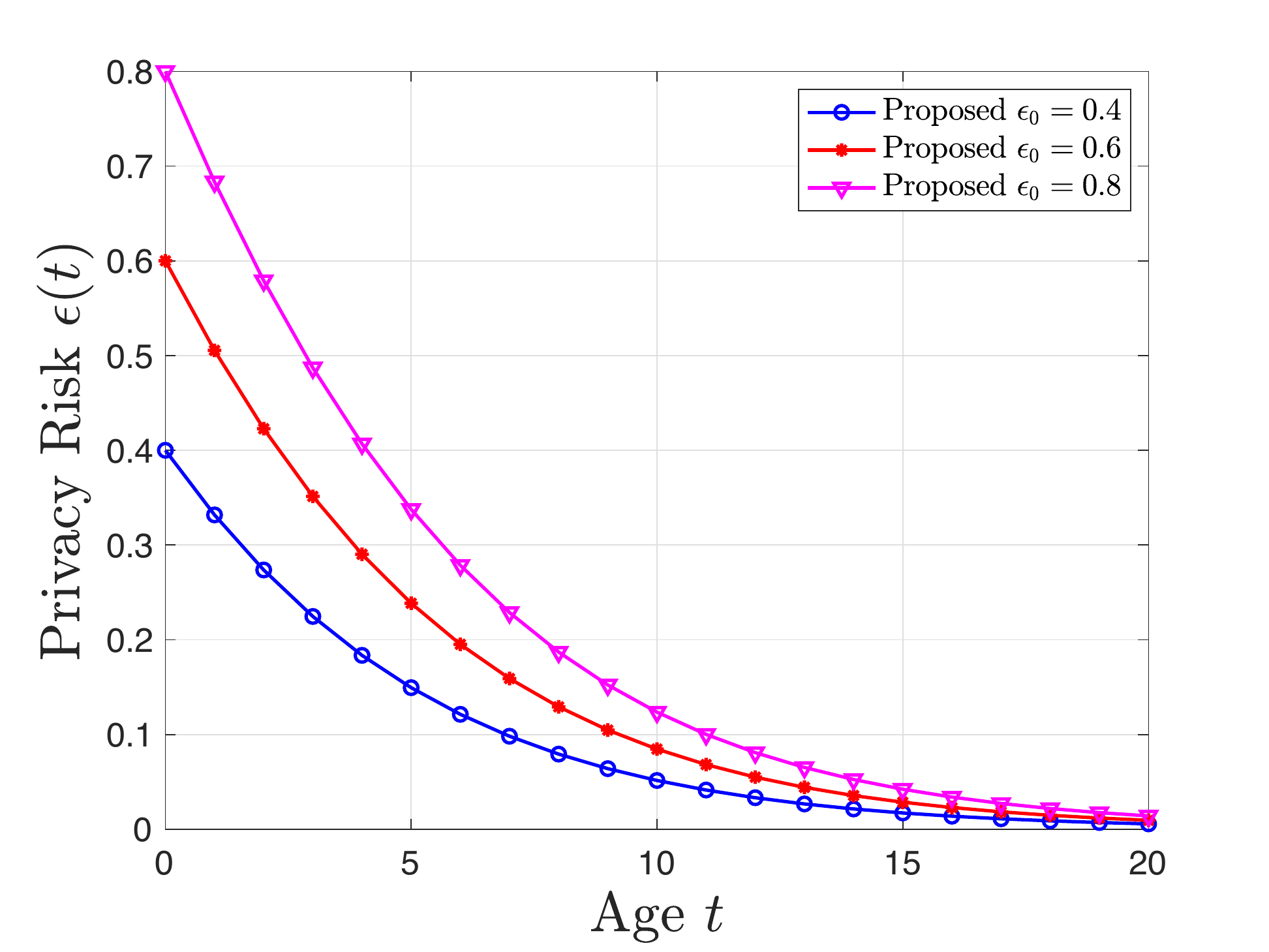}}
		\subfigure[]{\includegraphics[scale=.27]{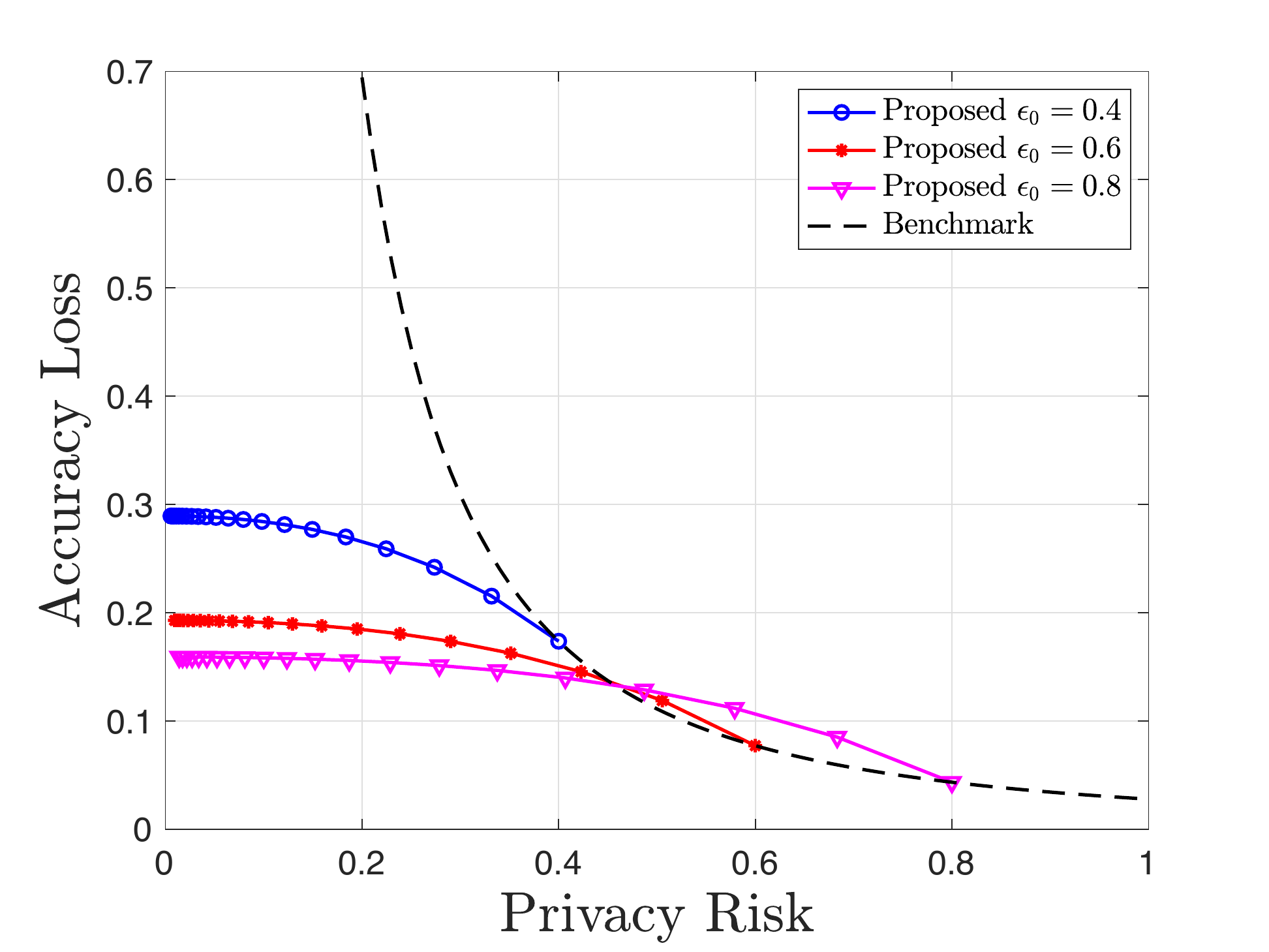}}
		\vspace{-5pt}
 		\caption{Numerical results for the autoregressive model.} \label{figeg3}
		\vspace{-10pt}
	\end{figure*}

We present numerical results in Figure \ref{figeg3} which are similar to those in Figure \ref{figeg1}.
In Figure \ref{figeg3}(a), we show that the accuracy loss increases in $t$, and for a given $\epsilon_C$, it converges to an upper bound as $t\rightarrow \infty$. 
In Figure \ref{figeg3}(b), we compare our proposed mechanisms against a classical DP benchmark that only injects Laplace noise.
To achieve an arbitrary small privacy risk, the accuracy loss incurred by the benchmark grows unbounded. 
Compared to the benchmark, combining both aging and noise injection 
does not incur meaningful accuracy loss, which is upper bounded as the privacy risk approaches $0$. 
	
\end{document}